\documentclass{amsart}
\usepackage{graphicx,dsfont,psfrag,amssymb,url,endnotes}
\usepackage[labelsep=period,font=bf]{caption}
\usepackage{subcaption}
\usepackage[font=bf]{subfig}
\DeclareCaptionLabelSeparator{periodspace}{.\quad}
\newtheorem{thm}{Theorem}[section]

\newtheorem{lem}[thm]{Lemma}
\newtheorem{prop}[thm]{Proposition}
\theoremstyle{definition}

\theoremstyle{remark}

\numberwithin{equation}{section}
\newtheorem{example}[thm]{Example}

\newcommand{\set}[1]{\left\{#1\right\}}

\newcommand{\ass}{\stackrel{\textup{\tiny def}}{=}}
\newcommand{\trr}{\triangleright}
\newcommand{\rrt}{\triangleleft}
\newcommand{\brr}{\blacktriangleright}
\newcommand{\bra}{\left|}
\newcommand{\ket}{\right \rangle}
 
\newcommand{\kett}{\left \langle}
\newcommand{\brat}{\right|}
\newcommand{\Hash}{\begin{minipage}{8pt}\includegraphics[width=8pt]{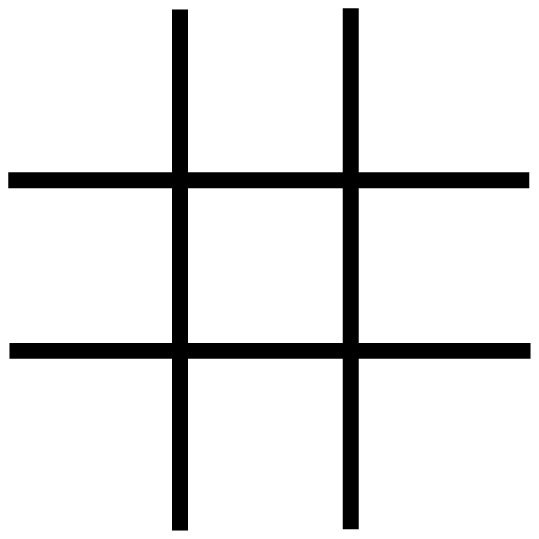}\end{minipage}}
\begin{document}

\title{Tales told by coloured tangles}%

\author{Daniel Moskovich and Avishy Y. Carmi}

\address{Faculty of Engineering Sciences \& Center for Quantum Information and Technology \\ Ben-Gurion University of the Negev, Beer-Sheva 8410501, Israel}



\begin{abstract}
Tangle machines are a topologically inspired diagrammatic formalism to describe information flow in networks. This paper begins with an expository account of tangle machines motivated by the problem of describing `covariance intersection' fusion of Gaussian estimators in networks. It then gives two examples in which tangle machines tell stories of adiabatic quantum computations, and discusses learning tangle machines from data.
\end{abstract}
\maketitle
\section{Introduction}


The Kantian explanation for `the unreasonable effectiveness of mathematics in the natural sciences'' is that the patterns which we seek out in nature are those patterns with which we are already familiar as innate categories of perception~\cite{Knuth:15}. When we encounter a new phenomenon, perhaps we first seek to describe this phenomenon in the category of perception most familiar to us. Taken one step further, perhaps human choice of research topics is itself influenced by the scientific language of the day.

The conventional graphical language in computer science is the language of flowcharts (\textit{e.g.}~\cite{Bohl:07}). A flowchart is a labeled directed graph in which edges represent transitions that can be concatenated. Thus, an ordered sequence of edges $(e_1,e_2,\ldots,e_k)$ in which the head of $e_i$ is the tail of $e_{i+1}$ for all $i=1,2,\ldots, k-1$ represents a transition from the tail of $e_1$ to the head of $e_k$. Edges can be prepended or appended to such sequences, and the transition represented by prepending $e_1$ to $(e_2,e_3,\ldots,e_k)$ is guaranteed to coincide with the transition represented by appending $e_k$ to $(e_1,e_2,\ldots,e_{k-1})$. This property is called \emph{associativity}. Thus, the language of flowcharts natively describes \emph{sequential} processes, although this paradigm may be expanded via a \emph{concurrency symbol}.

In previous work we proposed a different diagrammatic formalism called \emph{tangle machines}. A tangle machine looks visually similar to a coloured tangle diagram in low dimensional topology~\cite{CarmiMoskovich:14,CarmiMoskovich:14b,CarmiMoskovich:15a,CarmiMoskovich:15b,CarmiMoskovich:15c}. It consists of arcs which we think of as analogous to computer registers and of crossings that we call \emph{interactions} which we think of as analogous to logic gates. The interaction below involves three arcs called an \emph{agent}, an \emph{input patient}, and an \emph{output patient}. The interaction is labeled by a binary operation $\trr$, and for labels $x$ for the input patient and $y$ for the agent, the label of the output patient is $x\trr y$.

\begin{equation}
\psfrag{u}[r]{\small $x$}
\psfrag{y}[c]{\small $y$}
\psfrag{x}[l]{\small $x \trr y$}
\includegraphics[width=0.1\textwidth]{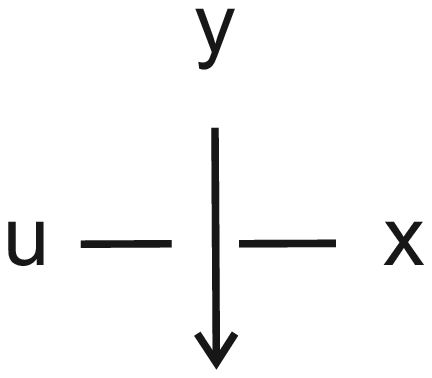}
\end{equation}
We call operation $\trr$ a \emph{fusion} if it satisfies three properties:

\begin{enumerate}
\item $x\trr x=x$ for all $x$. Thus, fusing information with itself should be the identity operation.
\item An input label can uniquely be recovered from its corresponding output label together with the agent label. Thus, fusion should `not forget'.
\item Fusing $x\trr z$ with $y\trr z$ should equal $(x\trr y)\trr z$. Thus, fusion should count the redundant appearance of $z$ only once towards the final result.
\end{enumerate}

The language of tangle machines natively describes networks of fusions (strengthening the requirement that the fusion `not forget' strengthened to $\bullet \trr y$ being a bijection). There is nothing inherently sequential about fusions, so tangle machines and flowcharts are quite different. What is the expressive power of tangle machines? More prosaically, what tales can a tangle machine tell?

Let us adopt an abstract view of a computation as a scheme to apply the information content of a computer program to the information content of input data in order to generate output data. Some computations (in this generalized sense) are fusions. It turns out that we can describe any Turing machine by a tangle machine under this paradigm. If we limit the number of interactions then we can describe any computation in the complexity class $\mathrm{IP}$~\cite{CarmiMoskovich:15a}. Parenthetically, we don't know how much more we can describe by tangle machines--- for example, can tangle machines with bounded interactions describe any computation in complexity class $\mathrm{MIP}$?

What more can be described by tangle machines? Archetypal mathematical examples of fusions are convex combinations and group conjugations, but we suggest that an archetypal real-world example of a fusion operation is fusion of estimates of a random variable. If the random variable follows a Bernoulli distribution then estimator fusion and convex combination essentially coincide (an example of this is in~\cite{CarmiMoskovich:15a}). We consider fusion of normally distributed estimates. The Kalman filter is a linear optimal measure to fuse Gaussian estimates, but it requires knowledge of how the estimate errors depend on one another. This information is not contained in the estimates themselves and may sometimes be entirely unknown. \emph{Covariance intersection} is a method to fuse Gaussian estimates whose error cross-correlations are unknown. Section~\ref{S:CI} introduces tangle machines as a diagrammatic language for information flow in networks in which information is fused by covariance intersection. It contains a geometric interpretation of covariance intersection as a choice of a point on a geodesic on a \emph{statistical manifold} whose points parameterize normal distributions.

Section~\ref{S:LDT} explains the connection of tangle machines to low dimensional topology. Section~\ref{SS:Knots} recalls a diagrammatic algebra from low dimensional topology. Section~\ref{SS:generalize} generalizes and discusses quandles, after which Section~\ref{SS:Info} defined tangle machines and shows how they may describe information fusion networks. The above sections represent an expository account of previous published work of the authors. As discussed in \cite{CarmiMoskovich:15c}, tangle machines (with colours suppressed) topologically arise as diagrams of embedded networks of intervals and spheres in standard Euclidean $4$--space.

Section~\ref{S:Quantum} takes tangle machines into the realm of quantum information, presenting several adiabatic quantum computations (\textsc{AQC}'s) with different performance features. Section~\ref{SS:TimeEntanglement} presents a computation in which a product state evolves into an entangles state, and Section~\ref{SS:2sat} presents a Boolean $2$--\textsc{SAT} problem over four qubits. A further example may be found in \cite[Section 5]{CarmiMoskovich:15b}.

Section~\ref{S:Utility} discusses potential utility of the tangle machine description. Inherited from low dimensional topology, tangle machines have topological invariants that are characteristic quantities giving high-level information about what is happening in the machine~\cite{CarmiMoskovich:14b,CarmiMoskovich:15c}. In particular, there are meaningful notions of \emph{capacity} and of \emph{complexity} for tangle machines. In addition, tangle machines provide a flexible description, and this flexibility may perhaps be useful for fault tolerance~\cite{CarmiMoskovich:14}.

Section~\ref{S:Learn} concludes by discussing the problem of adapting existing causality-detection algorithms to detect tangle machines. We illustrate with an example in which we detect a single crossing in real-world Google search data.

\begin{figure}
\centering
\psfrag{a}[r]{Covariance intersection}
\psfrag{b}[r]{Coloured tangle diagrams}
\psfrag{c}[r]{\raisebox{3pt}{Quantum information}\rule{0pt}{15pt}}
\psfrag{x}[c]{Tangle machines\phantom{BI}}
\psfrag{y}{}
\psfrag{z}[l]{\textsc{TM}s for \textsc{QI}}
\psfrag{t}[l]{Quantum info fusion}
\includegraphics[width=0.6\textwidth]{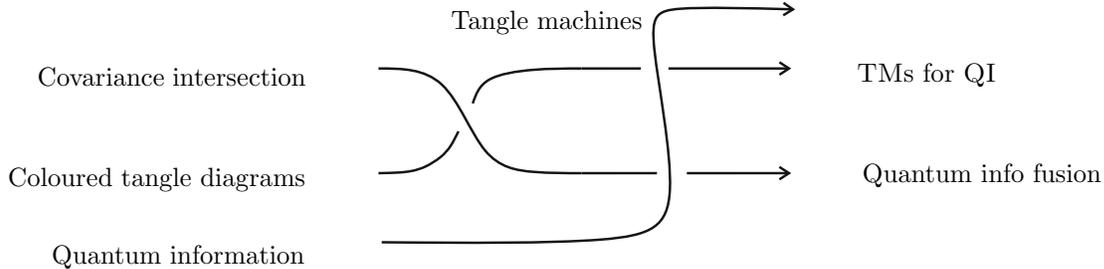}
\caption{\label{fig:r3plan} \small A coloured tangle representing this paper. At each crossing, the output topic is an application of the input topic to the agent (overcrossing) topic, perhaps in a different sense for each crossing. This figure suggests how a coloured tangle can tell a tale and has no rigourous mathematical meaning.}
\end{figure}

 The father of the use of tangle diagrams as a diagrammatic language for computation is Louis Kauffman, who represented automata, nonstandard set theory, and lambda calculus with tangle diagrams \cite{Kauffman:94,Kauffman:95,BuligaKauffman:13}. Buliga has suggested to represent computations using a different calculus of coloured tangles \cite{Buliga:11b}. In another direction, a different diagrammatic calculus, originating in higher category theory, has been used in the theory of quantum information--- see \textit{e.g.} \cite{AbramskyCoecke:09,BaezStay:11,Vicary:12}. In all of these approaches an element in a topologically inspired diagrammatic algebra represents a computation and equivalent diagrams represent bisimilar computations. Tangle machines are distinguished by the role of colours, by locality of orientations, and by binding together of operations in `interactions', as discussed in \cite{CarmiMoskovich:15b}.


\subsubsection*{Acknowledgement} D.M. was partially supported by the Helmsley Charitable Trust through the Agricultural, Biological and Cognitive Robotics Initiative of Ben-Gurion University of the Negev.

\section{An information geometric view of covariance intersection}\label{S:CI}

\subsection{Information geometry}\label{SS:InfoGeo}

Information geometry applies the methods of differential geometry to probability theory~\cite{AmariNagaoka:07}. Its main role so far has been to provide geometric coordinate-free elucidation to known facts. For example, the Cram\'{e}r--Rao inequality over $\mathds{R}$ which bounds variance from below by the reciprocal of Fisher information becomes a statement about the statistical manifold changing more in the vicinity of points of high curvature than in the vicinity of points of low curvature.

A \emph{statistical manifold} corresponding to an $n$--parameter family of distributions has a point for each possible value of the parameters. In particular, the statistical manifold for all normal distributions over $\mathds{R}$ is a surface whose coordinates are $(\mu,\sigma^2)$, the order and the variance of the distribution. A statistical manifold has a Riemannian metric called the \emph{Fisher information metric}. In the case of Gaussian random variables the metric is given by the inverse of the covariance matrix (over $\mathds{R}$ by the reciprocal of the variance as in the Cram\'{e}r--Rao bound).

We will be considering the problem of estimate fusion in an information geometric setting. The estimates will be points on the statistical manifold. Our main observation will be that that a fusion corresponds to a choice of a point on a geodesic.

\subsection{Estimator fusion}

In this section we ask the question: How can we fuse estimates of an unknown quantity whose interdependence is unknown?

Consider a pair of sensors which provide noisy observations about an unknown state vector $X$. Each sensor runs its own filtering algorithm to yield an estimate $X_{1,2}$ of $X$ and an estimate $C_{1,2}$ of the error covariance matrix of $X$.

\begin{multline}
\mathrm{cov}\left[ X-\hat{X}_i \right] \ass
E\left[(X-\hat{X}_i)(X-\hat{X}_i)^T \right] \\ - E\left[X-\hat{X}_i\right] E\left[X-\hat{X}_i\right]^T \quad \text{$i=0,1,2$}\enspace .
\end{multline}

We wish to optimally integrate these estimates $\left(\hat{X}_1,C_1\right)$ and $\left(\hat{X}_2,C_2\right)$. It is important that the fused estimate be \emph{consistent} (or \emph{conservative}), meaning that:
\begin{equation}
C_i  \geq \mathrm{cov}\left[ X-\hat{X}_i \right] \enspace ,
\end{equation}
\noindent \textit{i.e.} that the matrix difference $C_i  - \mathrm{cov}\left[ X-\hat{X}_i \right]$ is positive semi-definite. We would like all estimators to be consistent because inconsistent estimators may diverge and cause errors.

If the correlations between the estimate errors are known, a linear optimal fusion in the sense of minimum mean estimation error (\textsc{MMSE}) is given by the \emph{Kalman filter}. But such error cross-correlations are typically unknown and may be unmeasurable in real-world settings. Ignoring error cross-correlations can cause the Kalman filter estimates to be non-conservative and perhaps to diverge \cite{ChangChongMori:10}. The standard approach in applications is to increase the system noise artificially. But proper use of this heuristic requires substantial empirical analysis and compromises the integrity of the Kalman filter framework \cite{Niehsen:02}.


\emph{Covariance intersection} provides a method to fuse estimates whose error cross-correlations are unknown in a way that guarantees that the resulting estimate is consistent~\cite{JulierUhlman:01}. The covariance intersection method to fuse estimators $\left(\hat{X}_1,C_1\right)$ with $\left(\hat{X}_2,C_2\right)$ requires a choice of \emph{weight} $\omega\in (0,1)$. The bottleneck for practical covariance intersection is the computation of the optimal value for the weight $\omega$ with respect to some (typically nonlinear) cost function such as logdet or trace~\cite{Chen:02, Franken:05}. In the present paper we treat $\omega$ as a formal parameter or as an unknown constant for each pair of estimates to be fused.

We denote the covariance intersection of $\left(\hat{X}_1,C_1\right)$ with $\left(\hat{X}_2,C_2\right)$ with respect to $\omega\in (0,1)$ as follows:
\begin{equation}
\left(\hat{X}_a, C_a\right)_\omega\ass \left(\hat{X}_1, C_1\right)\triangleright_\omega \left(\hat{X}_2, C_2\right)\enspace .
\end{equation}
We construct $\left(\hat{X}_a, C_a\right)_\omega$ as follows:
\begin{subequations}
\label{eq:ci}
\begin{equation}
\hat{X}_{a} = (1-\omega) C_{a} C_1^{-1} \hat{X}_0 + \omega C_a C_2^{-1} \hat{X}_1 \enspace,
\end{equation}
\begin{equation}
C_a^{-1} = (1-\omega) C_1^{-1} + \omega C_2^{-1} \enspace .
\end{equation}
\end{subequations}

The working principle of covariance intersection is that it fuses two conservative estimates into a third conservative estimate. The reason, for Gaussian estimators, is that the \emph{covariance ellipse} of the fused estimator includes the intersection of the covariance ellipse of the estimators being fused (Figure~\ref{fig:cvinter}). A \emph{covariance ellipse} of a covariance matrix $C$ is the locus of vectors $v$ such that $v^{\, T} C^{-1} v\leq a$ where $a$ is some arbitrary (but fixed) constant. For Gaussian estimators, $C^{-1}$ represents the Fisher information. For details see~\cite{JulierUhlman:01}.

\begin{figure}
\centering
\psfrag{a}[c]{$C_1^{-1}$}
\psfrag{b}[c]{$C_{2}^{-1}$}
\psfrag{c}[l]{$C_a^{-1}=0.3C_1^{-1}+0.7C_2^{-1}$}
\includegraphics[width=0.6\textwidth]{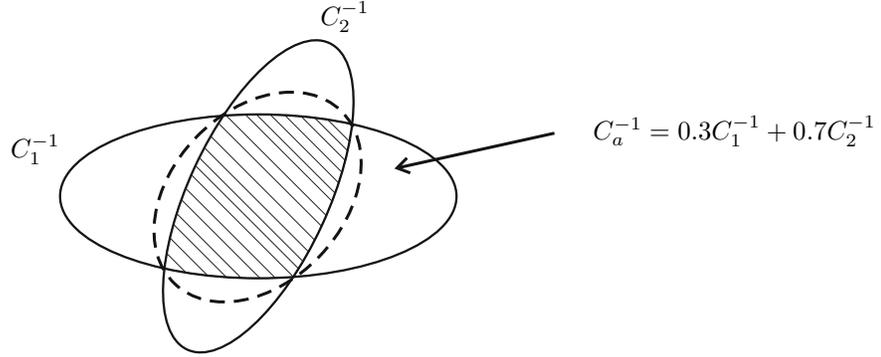}
\caption{\label{fig:cvinter} \small The covariance ellipse of the fused estimator is a minimal ellipse that includes within it the intersection of the covariance ellipses of the estimates being fused.}
\end{figure}

We prove the following proposition in the appendix.

\begin{prop}\label{P:Geodesic}
The $1$--parameter family of Gaussian covariance intersections $\left(\hat{X}_a, C_a\right)_{\omega\in [0,1]}$ parameterizes the geodesic from  $\left(\hat{X}_1,C_1\right)$ to $\left(\hat{X}_2,C_2\right)$.
\end{prop}

\begin{figure}
\centering
\psfrag{x}[c]{\small $p(x)$}
\psfrag{y}[l]{\small $p(y)$}
\psfrag{d}[c]{$\trr_\omega$}
\psfrag{z}[c]{\small $p(x \trr_\omega y)$}
\psfrag{a}[c]{$x$}
\psfrag{b}[c]{$y$}
\psfrag{c}[l]{$x \trr_\omega y$}
\psfrag{u}[c]{$\mu$}
\psfrag{v}[c]{$\Sigma$}
\psfrag{t}[c]{\emph{interaction}}
\psfrag{m}[c]{$M$}
\includegraphics[width=0.9\textwidth]{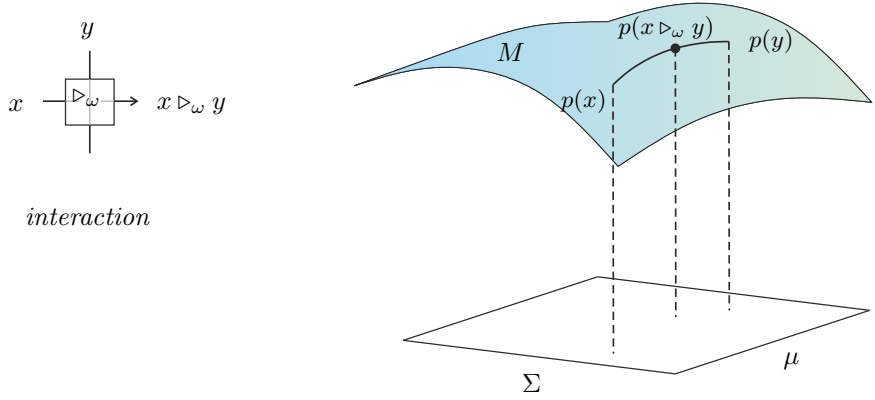}
\caption{\label{fig:manifold} \small A basic information fusion operation (left) and its geometric representation (right). Here $p(x)$ denotes a Gaussian distribution with mean $\hat{X}$ and covariance matrix $C$, which represents a point in the manifold.}
\end{figure}

\subsection{Properties of covariance intersection}\label{SS:CIProperties}

Covariance intersection has the following properties, which are evident from the geometric interpretation and can also be derived directly from~\eqref{eq:ci}:

\begin{description}
\item[Idempotence]
\[ a\trr_\omega a =a \qquad \text{for any estimator $a$ and for any $\omega\in(0,1)$.}\]
\item[Reversibility]
The function $\bullet \trr_\omega b$ from the statistical manifold to itself is a bijection. Thus, an estimate $a$ can uniquely be recovered from $b$, $a\trr_\omega b$, and $\omega$.
\item[No double counting] For any estimators $a$, $b$, and $c$ and for any weights $\omega,\omega^\prime\in (0,1)$, we have:
\begin{equation}
    (a\trr_{\omega^\prime} b)\trr_\omega c= (a\trr_\omega c)\trr_{\omega^\prime} (b\trr_\omega c) \enspace .
    \end{equation}
\noindent Thus, the redundant appearance of $c$ in $a\trr_\omega c$ and in $b\trr_\omega c$ counts only once towards the final result. This is the key property of covariance intersection that makes it insensitive to the dependence of different estimators.
\end{description}

\section{Coloured tangle diagrams in low dimensional topology}\label{S:LDT}

This section provides a new strand in our tale. Thus, we shall put aside all that came before, and we shall begin afresh.

\subsection{Colouring knots and tangles}\label{SS:Knots}

The fundamental problem in knot theory is to distinguish knots, which are smooth embeddings $K\colon\, S^1\rightarrow S^3$ of a directed circle into the $3$--sphere, considered up to an equivalence relation called \emph{ambient isotopy}. Identifying $S^3\simeq \mathds{R}^3\cup \{\infty\}$ and choosing a projection of the knot onto a generic plane in $\mathds{R}^3$ disjoint from the knot, we may represent $K$ as a \emph{knot diagram} such as one of those drawn in Figure~\ref{F:knot}. \emph{Reidemeister's Theorem} tells us that two knot diagrams represent ambient isotopic knots if and only if they are related by a finite sequence of the three \emph{Reidemeister Moves} of Figure~\ref{F:ReidemeisterMoves}.

\begin{figure}
\renewcommand{\thesubfigure}{\Alph{subfigure}}
\centering
\begin{subfigure}{.32\textwidth}
  \centering
  \psfrag{a}[c]{}
\includegraphics[width=0.67\textwidth]{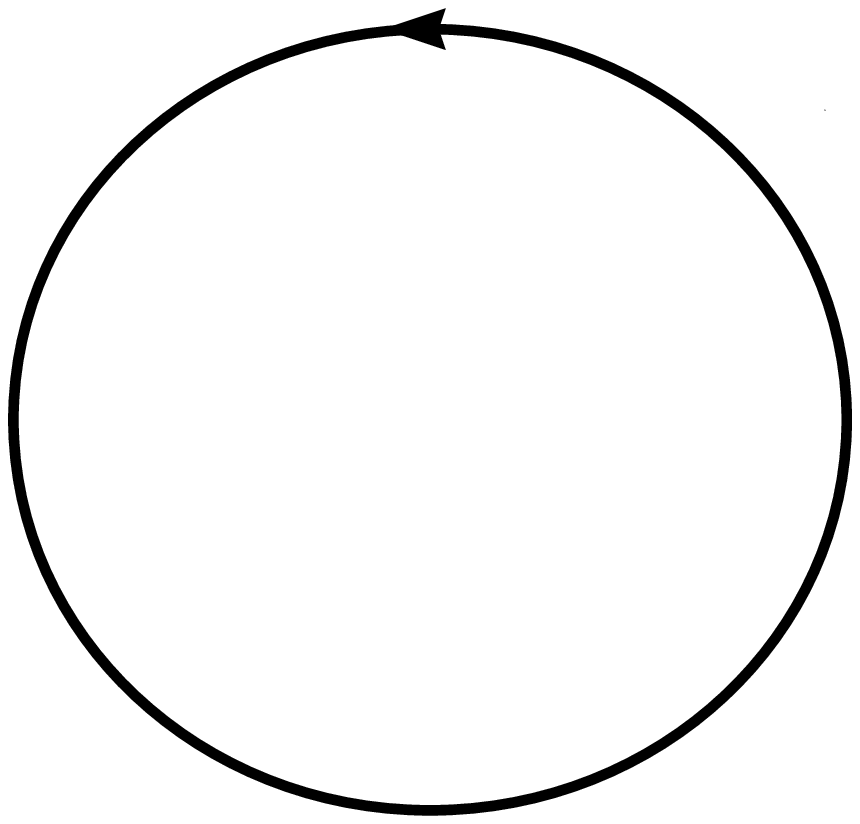}
\caption{\label{F:colknot01}}
\end{subfigure}%
\begin{subfigure}{.32\textwidth}
  \centering
  \psfrag{a}[c]{}
   \psfrag{b}[c]{}
  \psfrag{c}[c]{}
\includegraphics[width=0.8\textwidth]{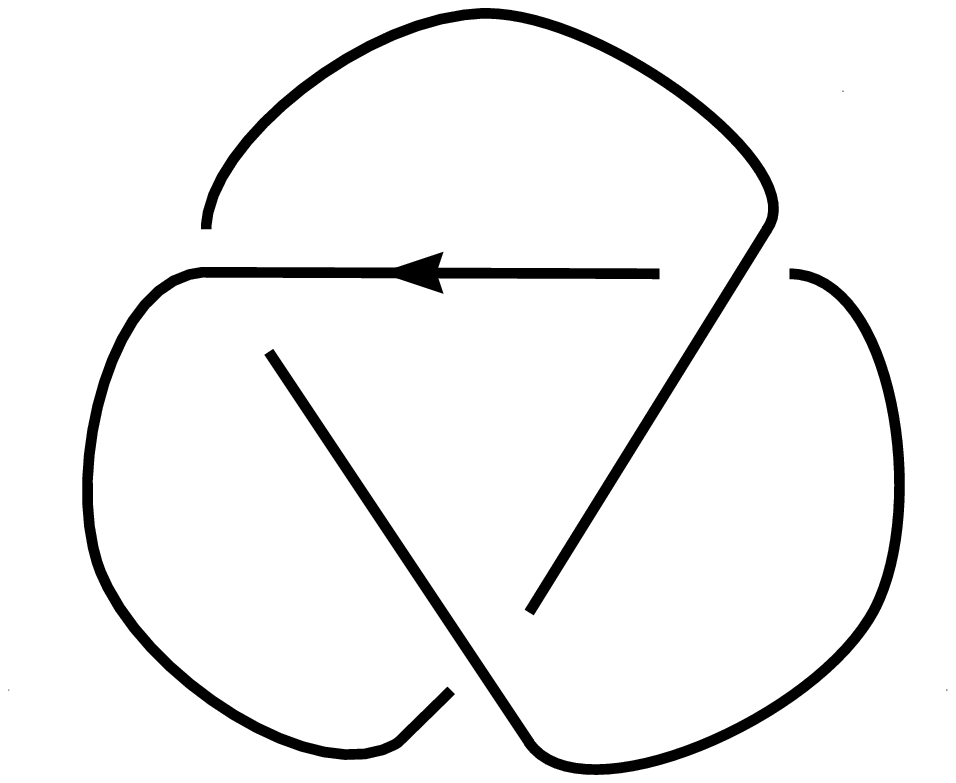}
\caption{\label{F:colknot02}}
\end{subfigure}
\begin{subfigure}{.32\textwidth}
  \centering
  \psfrag{1}[c]{}
   \psfrag{2}[c]{}
  \psfrag{3}[c]{}
  \psfrag{4}[c]{}
\includegraphics[width=0.8\textwidth]{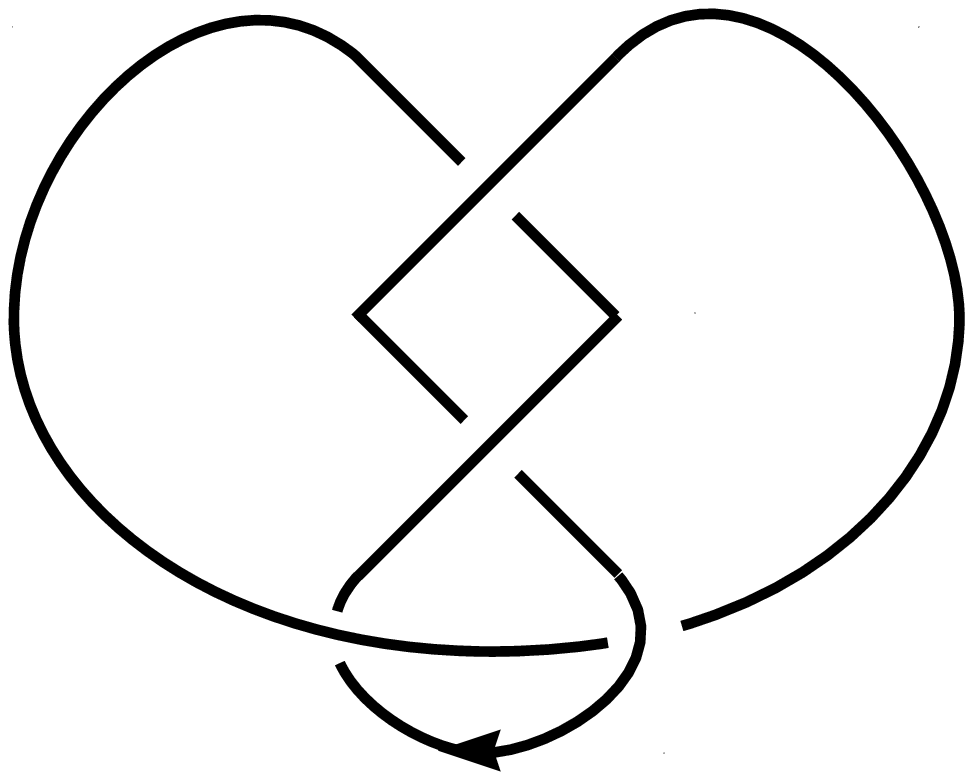}
\caption{\label{F:colknot03}}
\end{subfigure}
\caption{\label{F:knot} \small The unknot, the trefoil, and the figure-eight knot.}
\end{figure}

\begin{figure}
\psfrag{A}[r]{\small \emph{R2}}
\psfrag{B}[r]{\small \emph{R3}}
\psfrag{D}[r]{\small \emph{R1}}
\centering
\includegraphics[width=0.7\textwidth]{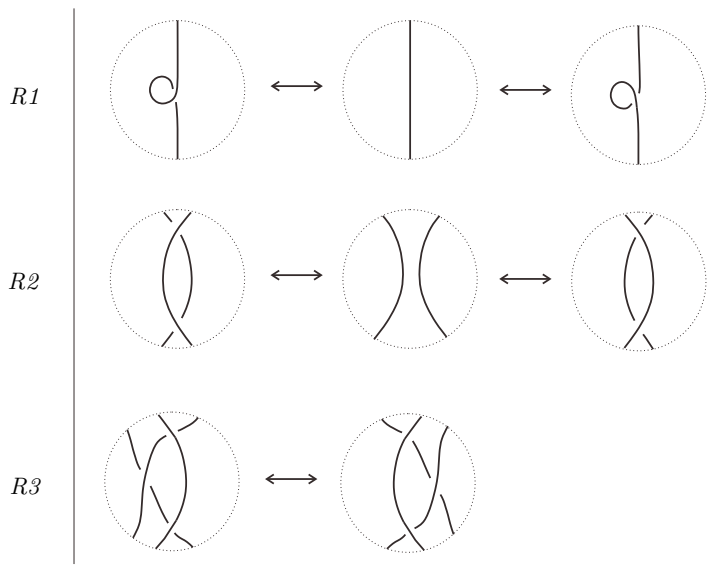}
\caption{\label{F:ReidemeisterMoves}The Reidemeister moves. Orientations are arbitrary. To execute a Reidemeister move, cut out a disc inside a tangle diagram containing one of the patterns above, and replace it with another disc containing the pattern on the other side of the Reidemeister move.}
\end{figure}

How are we to distinguish two non-ambient-isotopic knots? For example, how are we to distinguish the \emph{trefoil knot} from the \emph{unknot} in Figure~\ref{F:knot}? The original idea of Tietze~\cite{Tietze:08}, as explained later by Fox~\cite{Fox:61}, is to colour each arc in the knot diagram in one of three colours $\{\text{red, blue, green}\}$, represented numerically as $0$, $1$, and $2$, subject to the constraints that at least two different colours are used, and if two adjacent colours to a crossing are $i$ and $j$ then the third adjacent colour is $i\trr j \ass 2i-j\bmod 3$.

\begin{equation}
\psfrag{u}[r]{\small $x$}
\psfrag{y}[c]{\small $y$}
\psfrag{x}[l]{\small $x \trr y$}
\includegraphics[width=0.1\textwidth]{cross.eps}
\end{equation}

The trefoil can be coloured according to these rules (Figure~\ref{F:colknot}), but the unknot and the figure-eight knot cannot. This property of being \emph{$3$--colourable} is preserved under Reidemeister moves (Figure~\ref{F:ReidemeisterMoves1} will later show a more general statement), and thus we have shown that the trefoil is knotted and is not ambient isotopic to the figure-eight knot. But how to show that the figure-eight knot is knotted? Simple. Take five colours $0$, $1$, $2$, $3$, and $4$, and colour according to the same procedure. The figure-eight knot is $5$--colourable, but the trefoil and the unknot are not.

Generalizing, we arrive at the idea of defining an algebraic structure by which a knot can be coloured. Its elements are a set $Q$ and it comes equipped with an operation $\triangleright$ satisfying the following axioms:

\begin{description}
\item[Idempotence]
\[ x\trr x =x \qquad \forall \ x\in Q\enspace .\]
\item[Automorphism]
The function:
\begin{align}
\notag
\trr \, z \colon\ Q\  &\to \ Q\\
\notag
x\  &\mapsto \ x \trr z\enspace ,
\end{align}
\noindent is an automorphism of $(Q,\trr)$ for all $z\in Q$. In particular it is a bijection of sets and:
\begin{equation}
    (x\trr y)\trr z= (x\trr z)\trr (y\trr z) \qquad \forall x,y,z\in Q\enspace .
    \end{equation}
\end{description}

A structure $(Q,\triangleright)$ is called a \emph{quandle} (\textit{e.g.}~\cite{Joyce:82}). Several examples will be given in Section~\ref{SS:Info}.

\begin{figure}
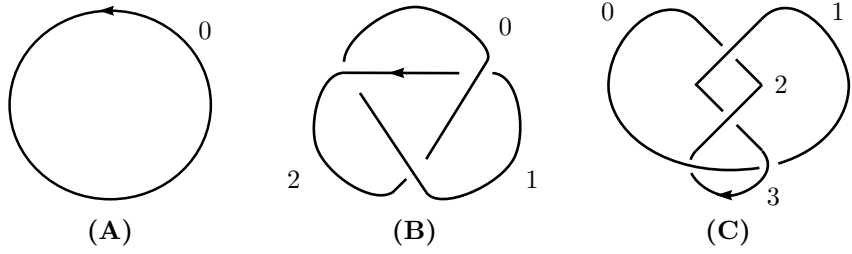

\renewcommand{\thesubfigure}{\Alph{subfigure}}
\centering
\begin{subfigure}{.32\textwidth}
  \centering
  \psfrag{a}[c]{$0$}
\includegraphics[width=0.67\textwidth]{unknotcol}
\caption{\label{F:colknot1}}
\end{subfigure}%
\begin{subfigure}{.32\textwidth}
  \centering
  \psfrag{a}[c]{$0$}
   \psfrag{b}[c]{$1$}
  \psfrag{c}[c]{$2$}
\includegraphics[width=0.8\textwidth]{trefoilcol}
\caption{\label{F:colknot2}}
\end{subfigure}
\begin{subfigure}{.32\textwidth}
  \centering
  \psfrag{1}[c]{$0$}
   \psfrag{2}[c]{$1$}
  \psfrag{3}[c]{$2$}
  \psfrag{4}[c]{$3$}
\includegraphics[width=0.8\textwidth]{figure8col}
\caption{\label{F:colknot3}}
\end{subfigure}
\caption{\label{F:colknot} \small A trivial colouring of the unknot, a $3$--colouring of the trefoil, and a $5$--colouring of the figure-eight knot. The quandle operation is $i\trr j = 2j-i \bmod p$.}
\end{figure}

\subsection{Generalizing coloured knots}\label{SS:generalize}

We did not use all of the properties of a knot when defining a quandle colouring. As far as colourings are concerned, several properties of a knot were irrelevant. Let us relax these, with a view to motivating tangle machines in Section~\ref{SS:Info}.

\begin{itemize}
\item We did not care about the number of connected components or whether they were closed. Quandle colourings are defined for \emph{tangles}, that are \emph{concatenations} of crossings, \textit{i.e.} shapes composed by scattering crossings in the plane and connecting up endpoints via disjoint line segments (Figure~\ref{F:comb} shows how a knot diagram can be constructed this way).Two tangles are equivalent if they are related by Reidemeister moves. See Figure~\ref{F:tangle}.
\item  We did not care about planarity. In terms of colourings, nothing is lost by allowing the connecting segments to intersect transversely. Relaxing the planarity requirement gives rise to \emph{virtual tangles}.
\item The orientations of the arcs were only important to give directionality to the quandle operation--- to specify which the `input underarc' is, versus the `output underarc'. There is no reason for directions of overcrossing arcs to match up; a global direction on an arc plays no role. Restricting orientations to overcrossing arcs and making them local introduces `disorientations' into our virtual tangles (\textit{e.g.}~\cite{ClarkMorrisonWalker:09})
\end{itemize}

\begin{figure}
\includegraphics[width=0.7\textwidth]{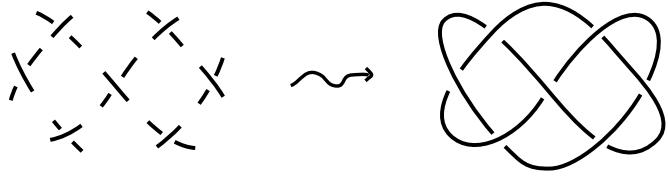}
\caption{\label{F:comb} A knot formed by concatenating crossings. Orientations ignored.}
\end{figure}

\begin{figure}
\renewcommand{\thesubfigure}{\Alph{subfigure}}
\centering
\begin{subfigure}{.41\textwidth}
\psfrag{a}{}\psfrag{d}{}
  \centering
\includegraphics[width=\textwidth]{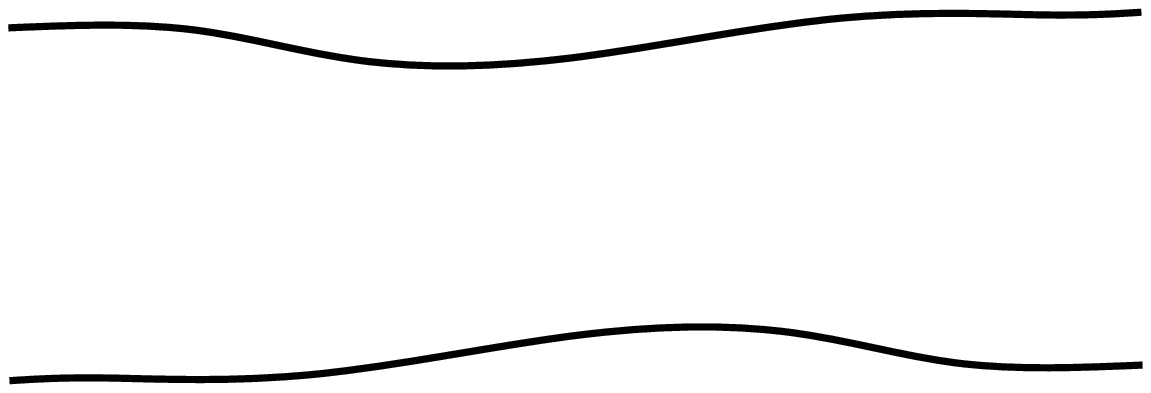}
\caption{\label{F:tangle1}}
\end{subfigure}%
\quad\quad\quad
\begin{subfigure}{.41\textwidth}
\psfrag{g}{}\psfrag{h}{}
  \centering
\includegraphics[width=\textwidth]{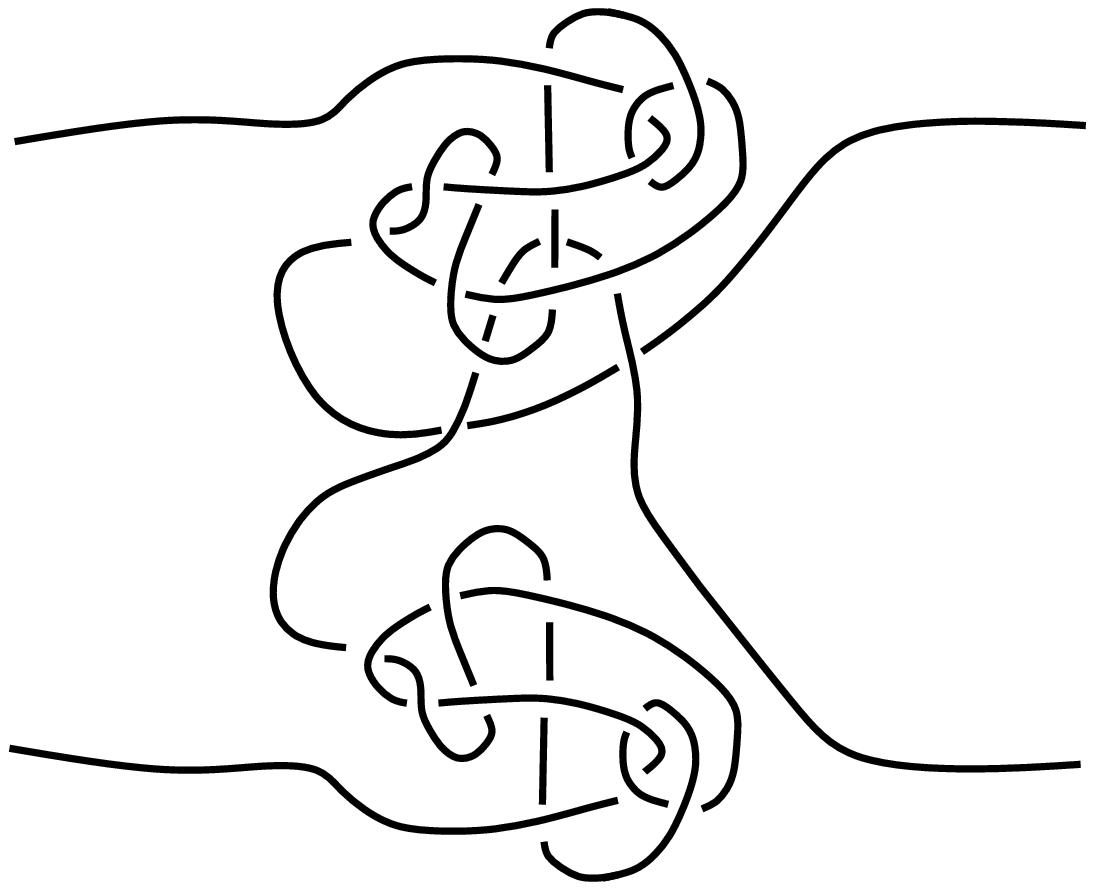}
\caption{\label{F:tangle2}}
\end{subfigure}
\caption{\label{F:tangle} \small Two equivalent tangles.}
\end{figure}

We will also want to consider quandles in a more general sense, in which $Q$ comes equipped with not one, but an entire family $B$ of binary operations, subject to the following axioms:

\begin{description}
\item[Idempotence]
\[ x\trr x =x \qquad \forall \ x\in Q \quad \forall \ \trr\in B\enspace .\]
\item[Automorphism]
The function:
\begin{align}
\notag
\trr \, z \colon\ Q\  &\to \ Q\\
\notag
x\  &\mapsto \ x \trr z\enspace ,
\end{align}
\noindent is an automorphism of $(Q,B)$ for all $z\in Q$ and for all $\trr\in B$. In particular each such function is a bijection of sets, and:
\begin{equation}
    (x\brr y)\trr z= (x\trr z)\brr (y\trr z) \qquad \forall\ x,y,z\in Q \quad\forall \trr,\brr\in B\enspace .
    \end{equation}
\end{description}

Such structures, perhaps subject to extra axioms which hold in our cases of interest, are variously called \emph{distributive $\Gamma$--idempotent right quasigroups}~\cite{Buliga:11a}, \emph{$G$--family of quandles} \cite{Ishii:12}, and \emph{multiquandles}~ \cite{Przytycki:11}. We continue to call them quandles.

We list some examples of quandles:

\begin{example}[Conjugation quandle]
Colours might be elements of a group $\Gamma$, and the operation might be conjugation:
\begin{equation}
x\brr y \ass y^{-1}xy\enspace .
\end{equation}
The pair $(\Gamma,\set{\brr})$ is called a \emph{conjugation quandle}. If the group $\Gamma$ is a dihedral group and we colour by reflections, we recover the notion of an $n$--colouring of a knot. \textit{e.g.}~\cite{Fox:61}.
\end{example}

\begin{example}[Linear quandle]\label{E:LinearQuandle}
Colours might be elements of a real vector space $Q$ and the operations might be convex combinations:
\begin{equation}
x\trr_\omega y \ass (1-\omega)x + \omega y \qquad \omega\in D\subseteq \mathds{R}\setminus\set{1}\enspace.
\end{equation}
The pair $\left(Q,\set{\trr_\omega}_{\omega\in D}\right)$ is called a \emph{linear quandle}.
\end{example}

\begin{example}[Loglinear quandle]\label{E:LogLinear}
In the same setting as Example~\ref{E:LinearQuandle}, consider the operations:
\begin{equation}
x\,\bar{\trr}_\omega\, y \ass x^{1-\omega}y^{\omega} \qquad \omega\in D\subseteq \mathds{R}\setminus\set{1}\enspace .
\end{equation}
The pair $\left(Q,\set{\bar{\trr}_\omega}_{\omega\in D}\right)$ is called a \emph{loglinear quandle}.
\end{example}



\subsection{Tangle machine definition}\label{SS:Info}

The attentive reader will have noticed that the properties of covariance intersection are similar to the quandle axioms. As we have seen, the set of points of a statistical manifold of Gaussian distributions satisfies the quandle axioms with respect to covariance intersection `choose a point on a geodesic' operations. In this section, we define tangle machines and colour them by estimators. Quandle operations will be covariance intersections.

The basic building block of tangle machine is called an \emph{interaction} (Figure~\ref{F:kebab}). It involves one strand coloured $y$ which we call a \emph{agent} and strands coloured $x_1,x_2,\ldots,x_k$ which we call \emph{patients}. We draw $y$ as a thick horizontal line, and we label line segments above $y$ by $x_1,x_2,\ldots,x_k$ correspondingly. An interaction has a weight $s$ associated to it, and the strands below $y$ are labeled by $x_1\trr_s y,x_2\trr_s y,\ldots x_k\trr_s y$ correspondingly. We may orient the strand labeled $y$ or all strands labeled $x$ (these two conventions are interchangeable). Sometimes we may be sloppy and orient both, in which case the orientations of the patients should be ignored.

\begin{figure}
\renewcommand{\thesubfigure}{}
\centering
\begin{subfigure}{.35\textwidth}
\psfrag{a}[c]{\small $y$}
\psfrag{b}[c]{\small $x_1$}
\psfrag{c}[r]{\small $x_1 \trr_s y$}
\psfrag{d}[c]{\small $x_2$}
\psfrag{e}[l]{\small $x_2 \trr_s y$}
\psfrag{f}[c]{\small $x_k$}
\psfrag{g}[c]{\small $x_k \trr_s y$}
\psfrag{h}[c]{\small $y$}
\psfrag{s}[c]{\small $\trr_s$}
\includegraphics[width=\textwidth]{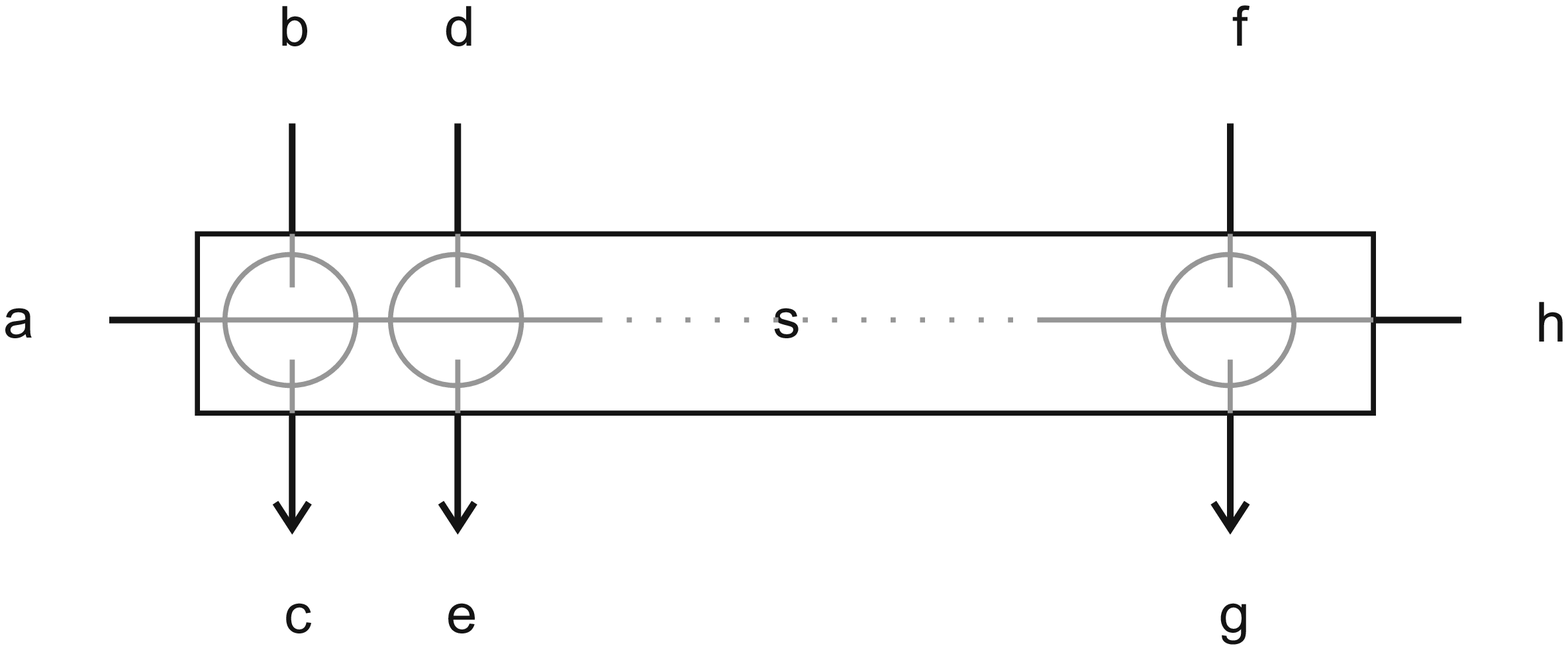}
\end{subfigure}\qquad\qquad\quad
\begin{subfigure}{.35\textwidth}
\includegraphics[width=\textwidth]{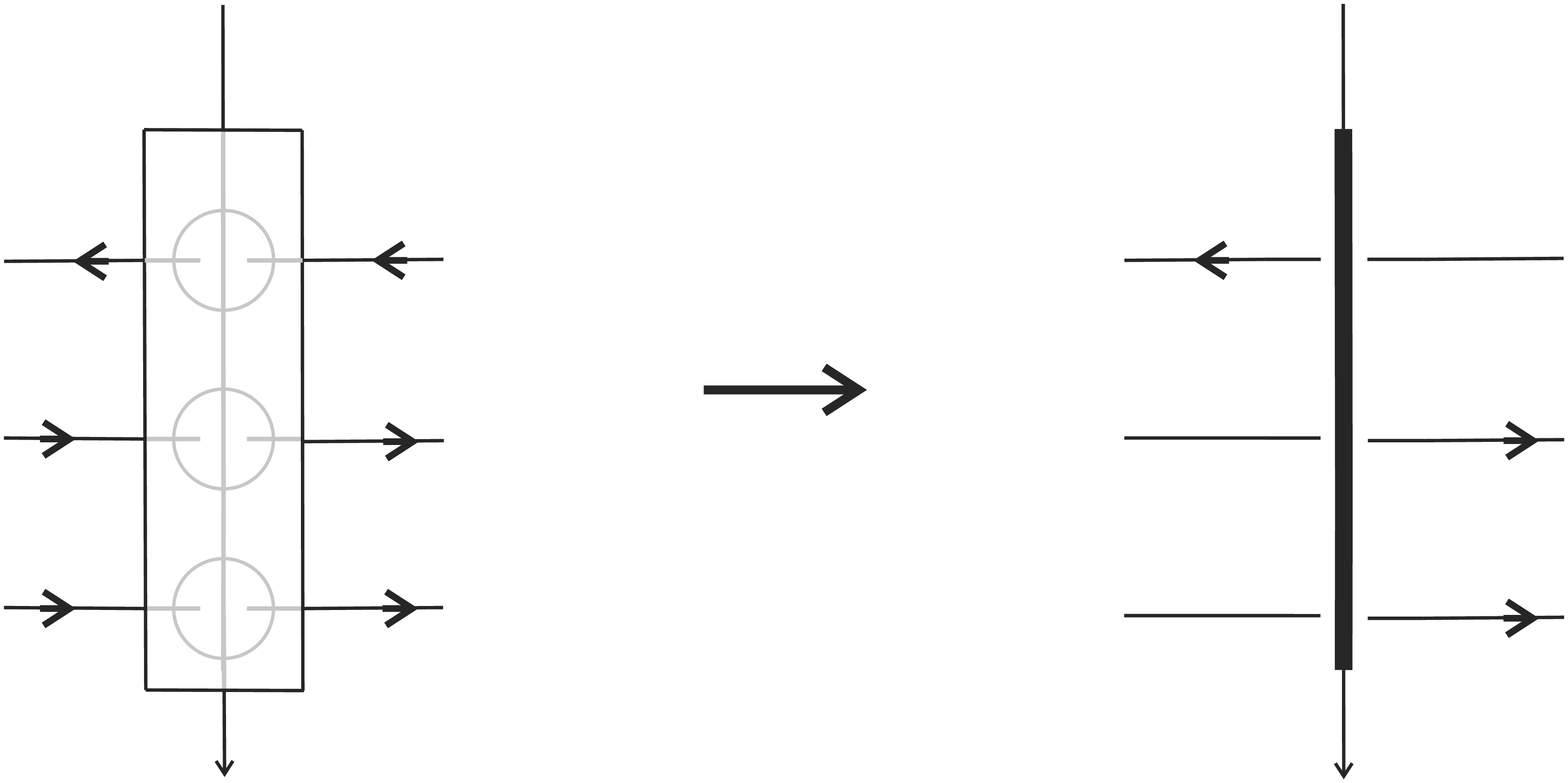}
\end{subfigure}
\caption{\label{F:kebab} An interaction. We usually draw agents as thick lines.}
\end{figure}

Interactions may be concatenated as shown in Figure~\ref{F:concat}. When labels are estimators, the concatenation of interactions represents a pair of related choices of points on geodesics, as shown in Figure~\ref{fig:manifold1}.

One reason that we allow multiple input patients for interactions is that the covariance intersection parameter $\omega$ is an estimated quantity. In general it will be different for different estimator fusions. When we know that it is the same, this information is important for us and we would like to record it. We typically represent agents with multiple agents by thickening the overarc.

\begin{figure}
\centering
\includegraphics[width=0.9\linewidth]{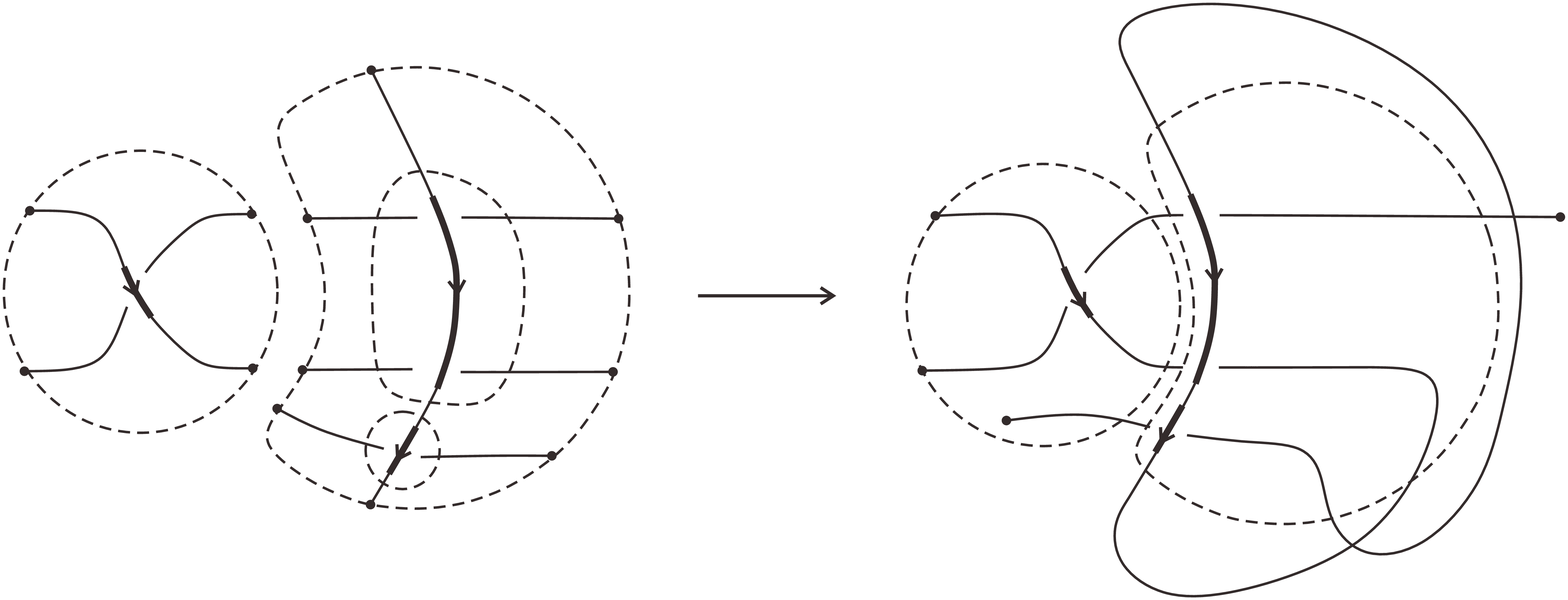}
\caption{\label{F:concat} Concatenation. Endpoints can be concatenated if they share the same colours (colours are suppressed in the above figure). At the third stage, note that in our formalism, only agents are oriented, and no compatibility requirement is imposed. The concatenation line chosen is arbitrary, and in particular it may intersect other concatenation lines. }
\end{figure}

\begin{figure}
\centering
\psfrag{x}[c]{\small $p(x)$}
\psfrag{y}[l]{\small $p(y)$}
\psfrag{f}[c]{\small $p(z)$}
\psfrag{h}[c]{\small $z$}
\psfrag{d}[c]{$\trr_s$}
\psfrag{e}[l]{$\trr_t$}
\psfrag{z}[c]{\small $p(x \trr_s y)$}
\psfrag{a}[c]{\small $x$}
\psfrag{b}[c]{\small $y$}
\psfrag{c}[c]{\small \; $x \trr_s y$}
\psfrag{k}[c]{\small $z \trr_t (x \trr_s y)$}
\psfrag{g}[c]{\small $p(z \trr_t (x \trr_s y))$}
\psfrag{u}[c]{$\mu$}
\psfrag{v}[c]{$\Sigma$}
\psfrag{t}[c]{\small \emph{tangle machine}}
\includegraphics[width=0.9\textwidth]{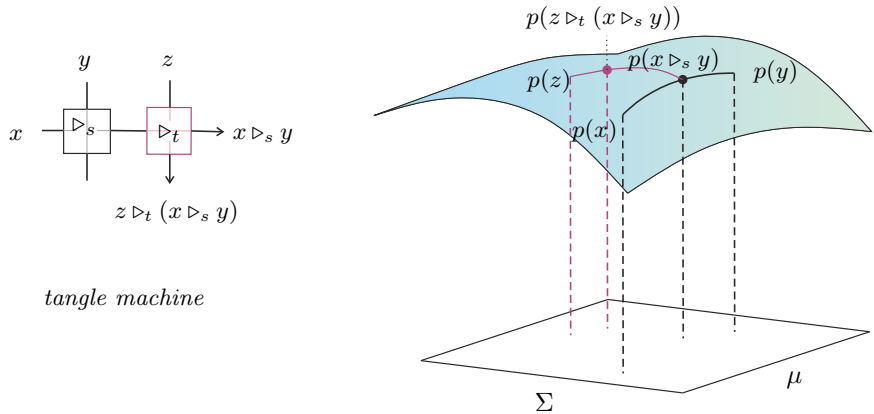}
\caption{\label{fig:manifold1} \small An information fusion network -- a tangle machine (left) and its geometric representation on a statistical manifold (right).}
\end{figure}

The lines used for concatenation do not matter, and the equivalence relation on diagrams generated by different choices of concatenating lines, and changes of local orientation, and adding or removing agents where `nothing happens', is illustrated in Figure~\ref{F:local_moves_machines}. The first two moves allow us to ignore orientations of agents when they don't matter.

Reidemeister moves for machines are illustrated in Figure~\ref{F:ReidemeisterMoves1}, where the general form of R3 is obtained in Figure~\ref{F:R3-box} where we first replace a box by a thickened line for convenience. We follow this convention for the remainder of the paper.

\begin{figure*}
\renewcommand{\thesubfigure}{R\arabic{subfigure}}
\begin{subfigure}{.4\textwidth}
  \centering
  \psfrag{s}[c]{$\trr_s$}
  \psfrag{a}[c]{$x$}
  \psfrag{x}[c]{\underline{$x$}}
  \psfrag{b}[c]{\underline{$x\trr_s x$}}
\includegraphics[width=0.8\textwidth]{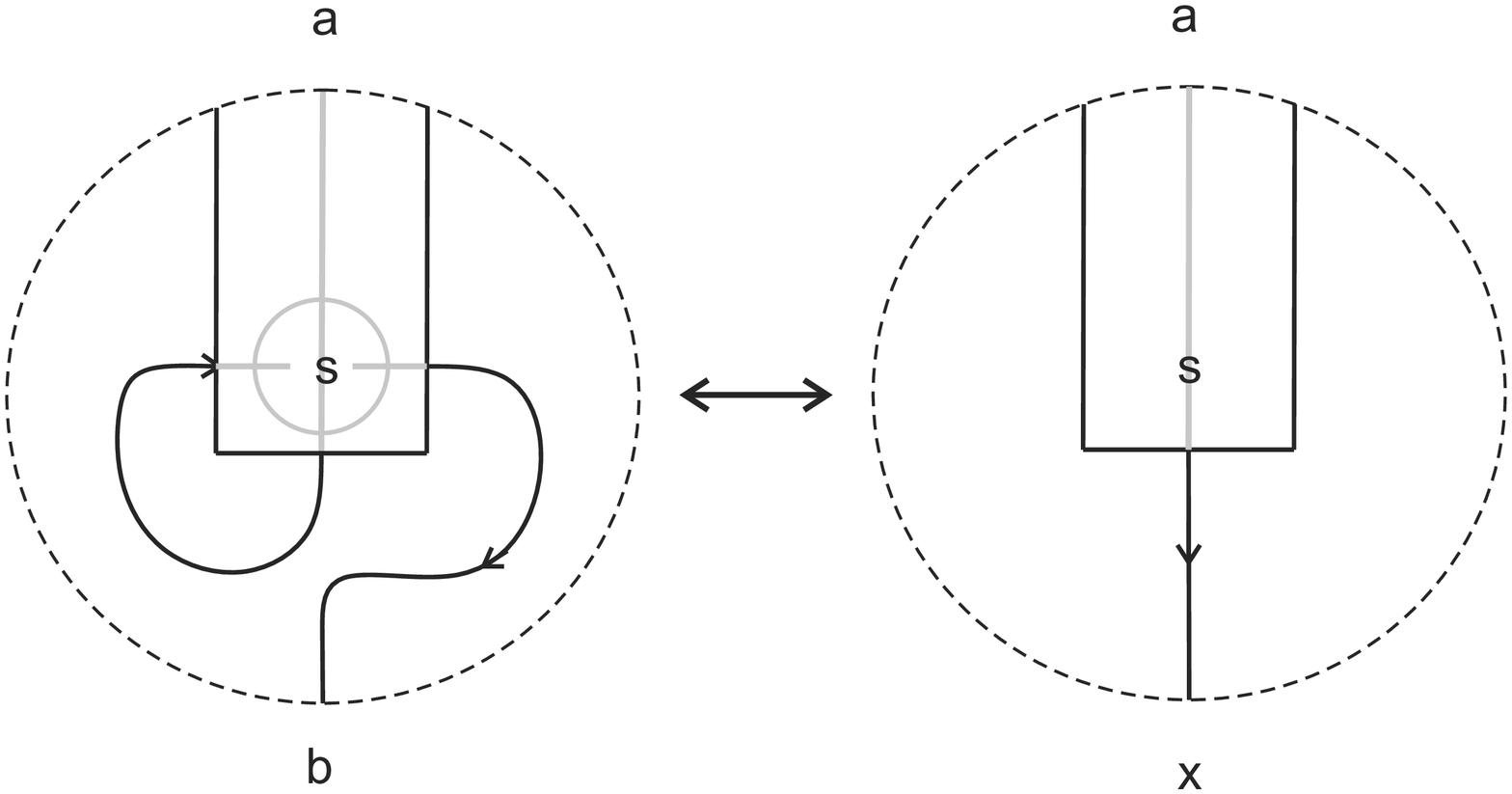}%
\caption{\label{F:r1}}
\end{subfigure}%
\quad\quad\quad
\begin{subfigure}{.5\textwidth}
  \centering
  \psfrag{c}[c]{$y$}\psfrag{a}[c]{$x$}\psfrag{x}[c]{$x$}\psfrag{d}[r]{\underline{$(x \trr_s y) \rrt_s y$}}\psfrag{b}[ld]{$x\trr_s y$}\psfrag{m}[c]{\underline{$a$}}
  \psfrag{s}[c]{$\trr_s$}
\includegraphics[width=0.8\textwidth]{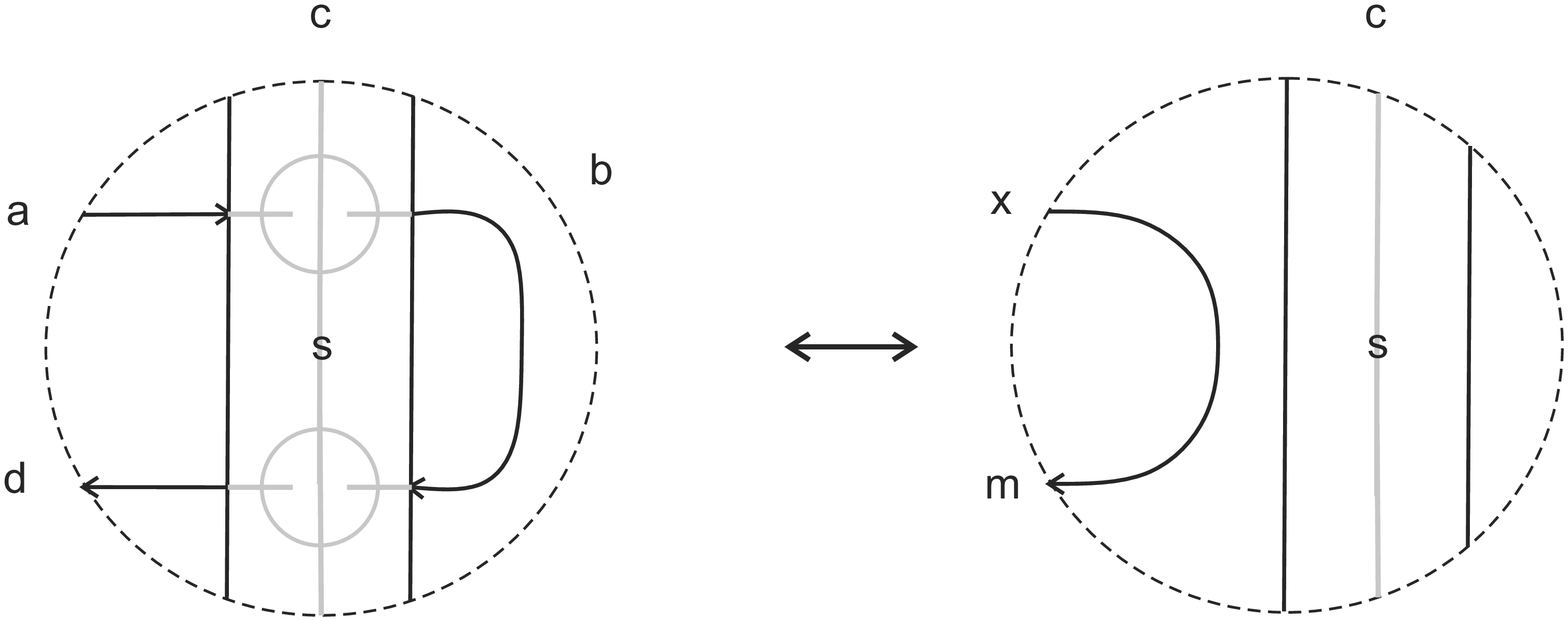}
\caption{\label{F:r2}}
\end{subfigure}%
\\[0.5cm]
\scalebox{0.99}{\begin{subfigure}{.95\textwidth}
  \centering
  \psfrag{s}[c]{$\trr_s$}
\psfrag{t}[c]{$\trr_t$}
\psfrag{a}[ld]{\underline{$(x\trr_s y)\trr_t z$}}\psfrag{b}[l]{$y\trr_t z$}
\psfrag{1}[c]{$y$}\psfrag{0}[c]{$x$}\psfrag{2}[c]{$z$}\psfrag{c}[ld]{\underline{$(x\trr_t z)\trr_s(y\trr_t z)$}}\psfrag{d}[l]{$y\trr_t z$}
\scalebox{0.9}{\includegraphics[width=0.8\textwidth]{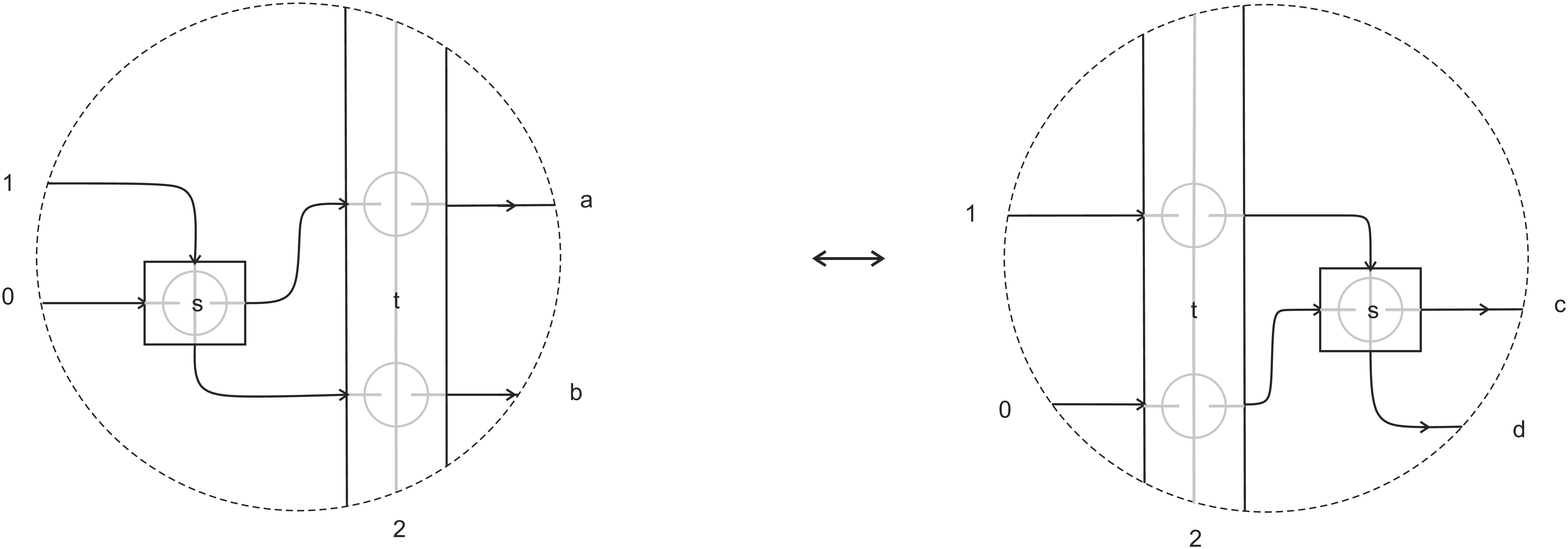}}
\caption{\label{F:r3}}
\end{subfigure}}
\caption{\label{F:ReidemeisterMoves1} \small The Reidemeister Moves, which are local modifications of information fusion networks. The moves are considered for all orientations on all edges. The properties of information fusion guarantees that the underlined colours on the \textsc{LHS} and on the \textsc{RHS} are equal and thus that these local moves are well-defined. In the R2 move, $\rrt$ represents the implicitly defined inverse operation to $\trr$. Only a special case of R3 is illustrated above--- the general case is given by Figure~{\ref{F:R3-box}}.}
\end{figure*}

\begin{figure*}
\centering
\begin{minipage}{2in}
\psfrag{a}[c]{\textbf{(R3)}}
\includegraphics[width=2in]{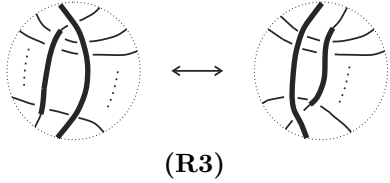}
\end{minipage}
\caption{\label{F:R3-box} \small The general R3 move.}
\end{figure*}

\begin{figure}[htb]
\centering
\psfrag{a}[c]{\small $\trr$}
\psfrag{b}[c]{\small $\rrt$}
\psfrag{V}[c]{\small \emph{I1}}
\psfrag{x}[c]{\small $x$}
\psfrag{T}[c]{\small \emph{VR1}}\psfrag{R}[c]{\small \emph{VR2}}\psfrag{S}[c]{\small \emph{VR3}}
\psfrag{Q}[c]{\small \emph{SV}}\psfrag{D}[c]{\small \emph{I2}}\psfrag{E}[c]{\small \emph{FM1}}\psfrag{F}[c]{\small \emph{FM2}}\psfrag{C}[c]{\small \emph{I3}}\psfrag{Y}[c]{\small \emph{ST}}
\psfrag{X}[c]{\small \emph{ST}}
\includegraphics[width=0.85\textwidth]{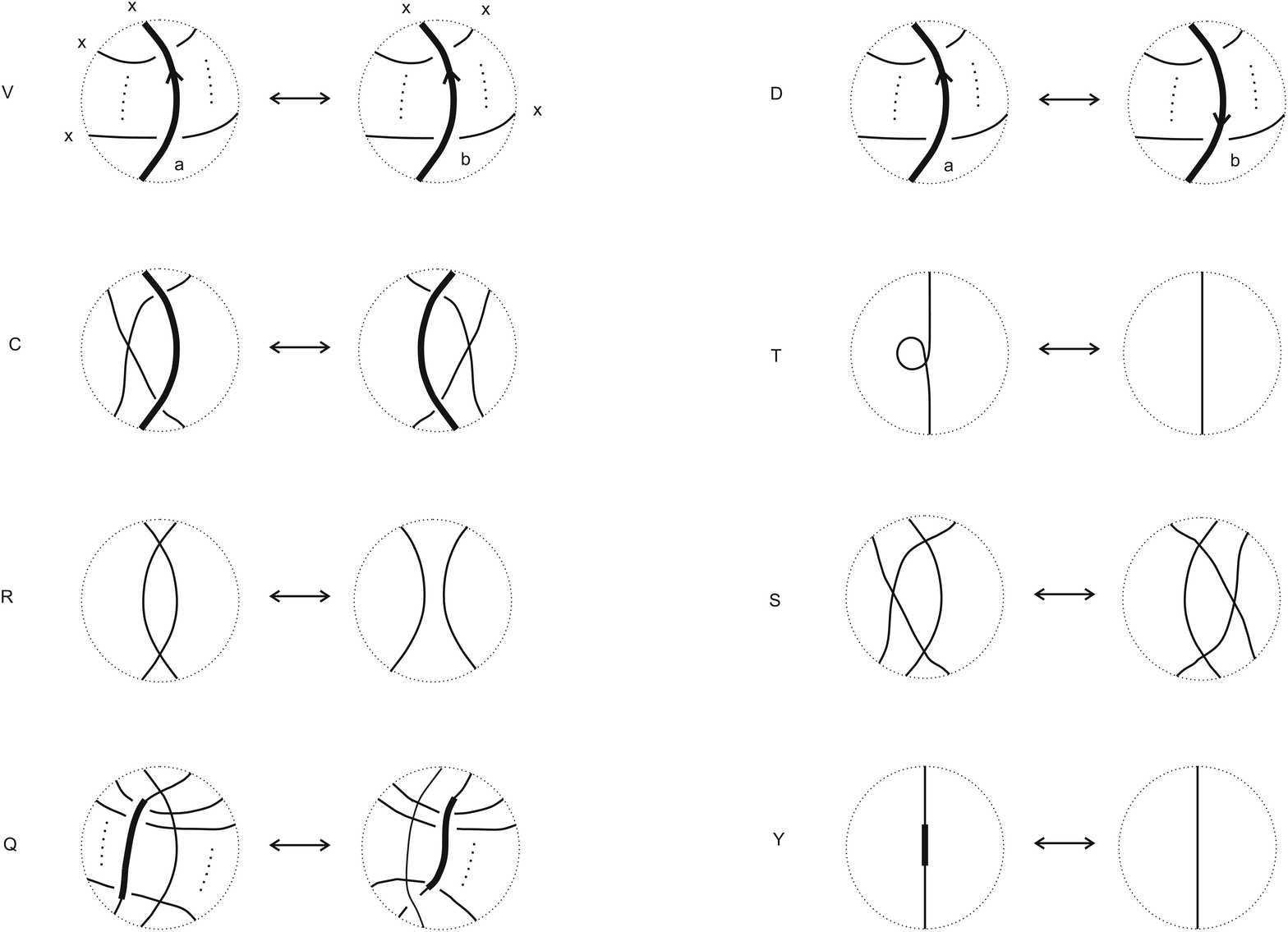}
\caption{\label{F:local_moves_machines} \small Cosmetic moves for machines. Where directions are not indicated, the meaning is that the move is valid for any directions, and the same for colourings. Here, $\rrt$ denotes the (implicitly defined) inverse operation to $\trr$.}
\end{figure}

One notion of tangle machine computation is that it is a choice of arcs as \emph{input registers} and another choice of arcs as \emph{output registers}. Local moves may not involve these registers. If the colours of the input registers uniquely determine the colours of the output registers, then we consider this determination as a computation. We may think of interactions as being analogues to that of logic gates. But because interaction concatenation is non-sequential (an interaction may even be concatenated to itself), this notion of computation is also non-sequential, and a tangle machine computation may be thought of as taking place `all at once' via an oracle.

When interactions represent choice of point on a geodesic between the input patient and the agent on a statistical manifold, machines represent a network of these.

\section{Tangle machines coloured by Hamiltonians}\label{S:Quantum}

\subsection{Tangle machines representing piecewise adiabatic quantum computations}

Quantum information is based on quantum probability theory, in which \emph{density matrices} take the place of classical probability distributions. A set of density matrices may be given the structure of a statistical manifold. A tangent vector to this manifold is a \emph{Hamiltonian matrix}, which physically describes the dynamics of a quantum system.

One research direction would be to colour tangle machines by density matrices and to fuse them using points on the geodesic between them on the statistical manifold. Instead, we colour our machines by Hamiltonians and we consider the linear quandle with operations indexed by a time parameter $s\in (0,1)$~\cite{CarmiMoskovich:15b}.

Adiabatic quantum computation (\textsc{AQC}) considers evolving Hamiltonians of the form $(1-s)H_0 + sH_1$  to perform quantum computation~\cite{Farhi:2000}.
Analogously to~\eqref{eq:ci} in which the statistical moments of two estimators are fused, the interaction with output:
\begin{equation}
\label{eq:ev}
H_s = H_0 \trr_s H_1 \ass (1-s) H_0 + s H_1 \enspace ,
\end{equation}
may be viewed as a fusion of two statistical ensembles underlying two quantum systems. As before the parameter $s \in (0,1)$ determines the contribution of
each component to the new ensemble, $H_s$.

The quantum mechanical counterpart of independence of estimators is that the Hamiltonians commute, $[H_0, H_1] = 0$, in which case the
underlying density matrices satisfy:
\begin{equation}
\rho_s \propto \exp\left(-\beta H_s\right) = \exp\left(- \beta H_0\right)^{1-s}
\exp\left(-\beta H_1 \right)^s \propto \rho_0^{1-s} \rho_1^s\enspace .
\end{equation}
\noindent This is analogous to the formula for a geodesic on a statistical manifold in Section~\ref{S:CI}.

The idea of \textsc{AQC} is to evolve a Hamiltonian $H_0$ representing a problem whose solution (the ground state) is easy, into a different and perhaps more complicated Hamiltonian $H_1$ whose ground state is the solution to the computational task of interest. The computation initializes the system in its ground state, the ground state of $H_0$, and then slowly evolves its Hamiltonian to $H_1$. This process is called \emph{quantum annealing}. By the Adiabatic Theorem, the system remains in its ground state throughout the evolution process, and the computation concludes at the ground state of $H_1$, that is the sought-after solution. The \emph{computing time} of an \textsc{AQC} is inversely proportional to the square of the minimal energy gap between the ground state and the rest of the spectrum, namely to the square of $g \ass \lambda_1 - \lambda_0$, where $\lambda_{i+1} \geq \lambda_i$ are the underlying energy eigenvalues of the Hamiltonian. Our objective is thus to maximize the minimal $g$ along the evolution path or to minimize the computation time.

As shown in a previous paper~\cite{CarmiMoskovich:15b} and as we will again illustrate below, a tangle machine may be used to represent a piecewise-defined evolution path of a Hamiltonian. It is well-known that nonlinear evolution paths speed up quantum computation \cite{Farhi:2002,Farhi:2011}. The topology of the machine itself is only of interest if we can measure its colours at intermediate times and at intermediate arcs. The technological barrier to implementation is the \emph{no-cloning theorem} which tells us that we cannot copy a state, and therefore we have lost it the moment we have measured it. From a future technological perspective we might envision a solution in which concatenation of interactions occurs by means of \emph{quantum teleportation} of an entire Hamiltonian from one interaction to the next. In such a case the tangle machine representation of an \textsc{AQC} would have all of the same features as a tangle machine representation of a classical information fusion network. Without such technology, the tangle machines below merely tell us a tale of Hamiltonian evolution.

We define one additional piece of structure. Choose a \emph{time deformation parameter} $\alpha \in (0,1)$ and define the normalized time variables:
\begin{equation}
t_i = \min \{1, \alpha^{-i} t \}, \quad i=0,1,2,\ldots
\end{equation}
with $t \in (0,1)$. A time parameter $t_i$ is associated with the \emph{$i$-th time-frame}, and by convention all interactions within this `time-frame' (drawn as a region in the plane) are labeled $t_i$. For small $\alpha$, the computations in the $k$-th time-frame nearly terminate when the computations in the $k-1$ frame begin. In the limit $\alpha \to 0$ we recover sequential computation.

Our diagrams are drawn along a time axis showing the time-frames. Interactions appearing in a particular time-frame, say the $k$-th one, are understood to be executed using the respective time parameter, \textit{i.e.} $\trr_{t_k}$. In what follows we use the shorthand $\trr_k$ for $\trr_{t_k}$.

\subsection{Time-Entanglement asymmetry}\label{SS:TimeEntanglement}

In this section we describe a computation involving a single pair of qubits in which a product state evolves into an entangled state.
As we shall see, the computation time is proportional to the amount of entanglement.

We begin with a standard adiabatic computation
\begin{equation}
O(t) = H_0 \trr_0 H_1 \enspace ,
\end{equation}
where the initial and problem Hamiltonians are given by:
\begin{subequations}
\begin{equation}
H_0 = P_x^0 \trr_{\oplus} P_x^1 \ass (1-a) (P_x^0 \otimes 1) + a (1 \otimes P_x^1)\enspace ;
\end{equation}
\begin{equation}
H_1 = 1 - \left(\lambda \bra 00 \ket + \sqrt{1-\lambda^2} \bra 11 \ket\right)\left(\lambda \bra 00 \ket + \sqrt{1-\lambda^2} \bra 11 \ket\right)^{\dagger}\enspace ;
\end{equation}
\end{subequations}
where $a,\lambda \in (0,1)$ are fixed parameters. The \emph{projectors} are defined as:
\begin{equation}
P_x^j \ass \frac{1}{2}(1 - (-1)^j X)\enspace ,
\end{equation}
where $X$ is the Pauli \textsc{X} matrix. The initial ground state is the product $\bra 01 \ket_x$ whereas the final ground state is a state
(written above in the computational basis \textsc{Z}) whose entanglement entropy is
\[S(\lambda) = -\lambda^2 \log (\lambda^2) - (1-\lambda^2) \log (1-\lambda^2)\enspace .\]
\noindent A maximally entangled state is obtained for $\lambda = \frac{1}{\sqrt{2}}$.

Figure~\ref{fig:time_ent} graphs computation time against an increasing $\lambda$ for two different values of $a$.

\begin{figure}[htb]
\centering
\includegraphics[width=0.65\textwidth]{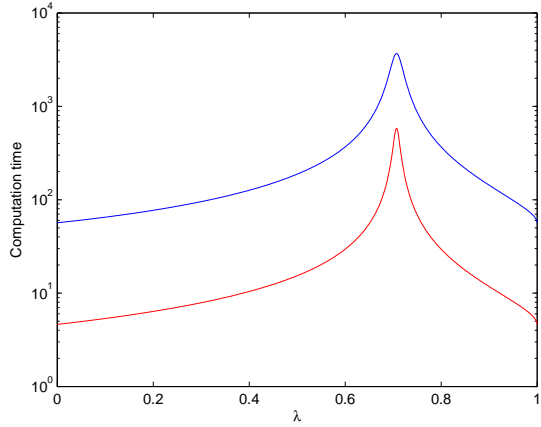}
\caption{\small Computation time against $\lambda$ for two values of $a$: (Blue) $a=0.95$, and (Red) $a=0.5$.}
\label{fig:time_ent}
\end{figure}

Let's solve the same problem using a tangle machine and time deformation. We choose time variables as:
\begin{equation}
t_1 \ass a \cdot \min\{1, \alpha^{-1} t\}, \qquad t_0 = t\enspace .
\end{equation}

Let us solve the same computational problem in two different ways, represented by the tangle machines in Figure~\ref{fig:entang}:
\begin{subequations}
\begin{equation}
O_1(t) = \left[(P_x^0 \otimes 1) \trr_1 (1 \otimes P_x^1)\right] \trr_0 H_1 \enspace ;
\end{equation}
\begin{equation}
O_1^\prime(t) = \left[(P_x^0 \otimes 1) \trr_1 H_1\right] \trr_0 \left[(1 \otimes P_x^1) \trr_1 H_1 \right]\enspace .
\end{equation}
\end{subequations}

The computation time is graphed with respect to $\lambda$ in Figure~\ref{fig:time_ent1}. The performance of the standard \textsc{AQC} from Figure~\ref{fig:time_ent} is also shown for comparison. Here $\alpha=0.5$.

\begin{figure}[htb]
\centering
\psfrag{0}[r]{\small $P_x^0 \otimes 1$}
\psfrag{1}[r]{\small $1 \otimes P_x^1$}
\psfrag{2}[c]{\small $H_1$}
\psfrag{a}[c]{\small $O_1$}
\psfrag{b}[c]{\small $O_2$}
\psfrag{c}[c]{\small $O_3$}
\psfrag{t}[c]{\small $t_1$}
\psfrag{s}[c]{\small $t_0$}
\psfrag{e}[c]{\small $O_1^\prime$}
\psfrag{f}[c]{\small $O_2^\prime$}
\psfrag{g}[c]{\small $O_3^\prime$}
\includegraphics[width=0.4\textwidth]{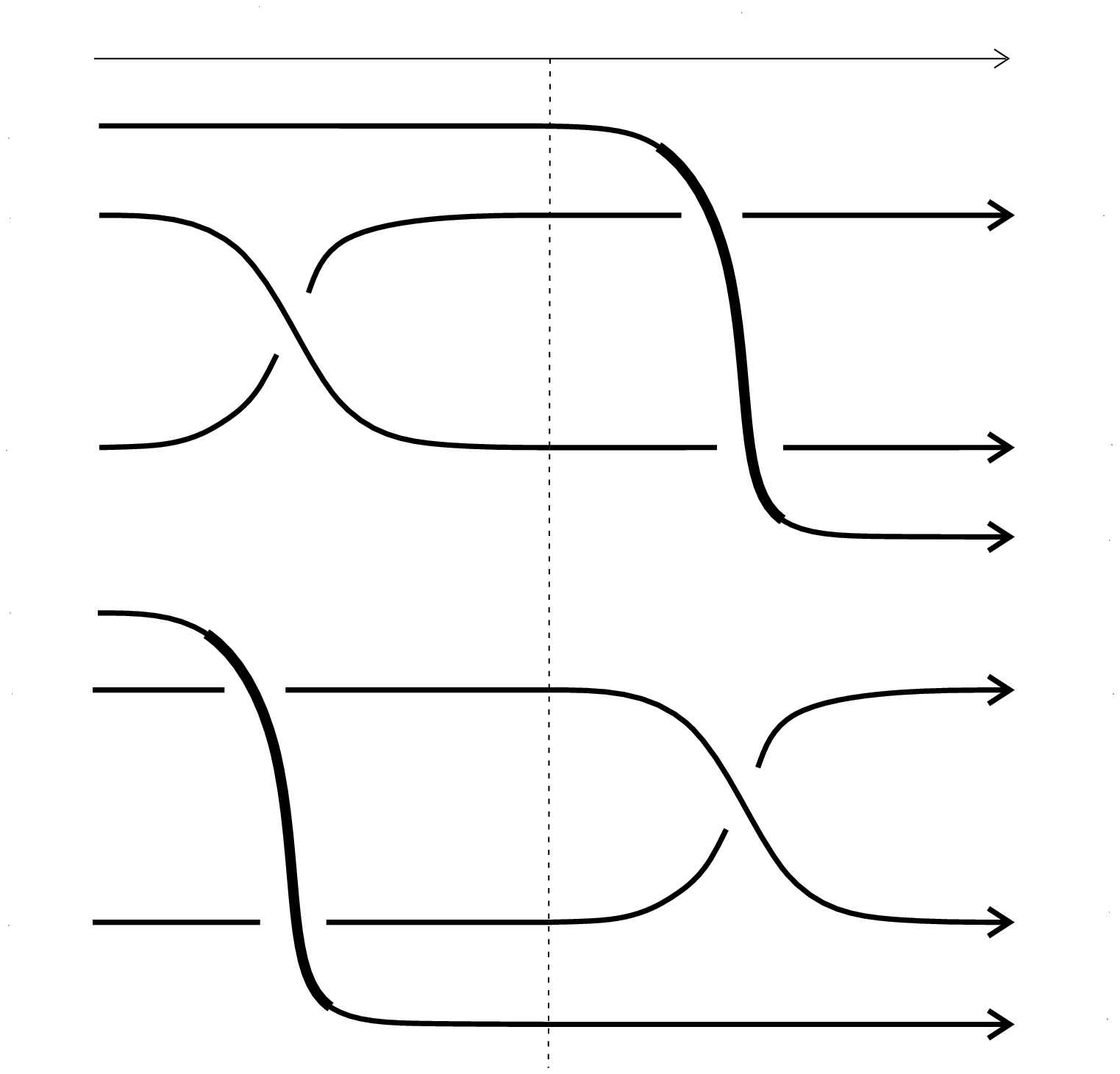}
\caption{\small Equivalent computations.}
\label{fig:entang}
\end{figure}

\begin{figure}[htb]
\centering
\includegraphics[width=0.65\textwidth]{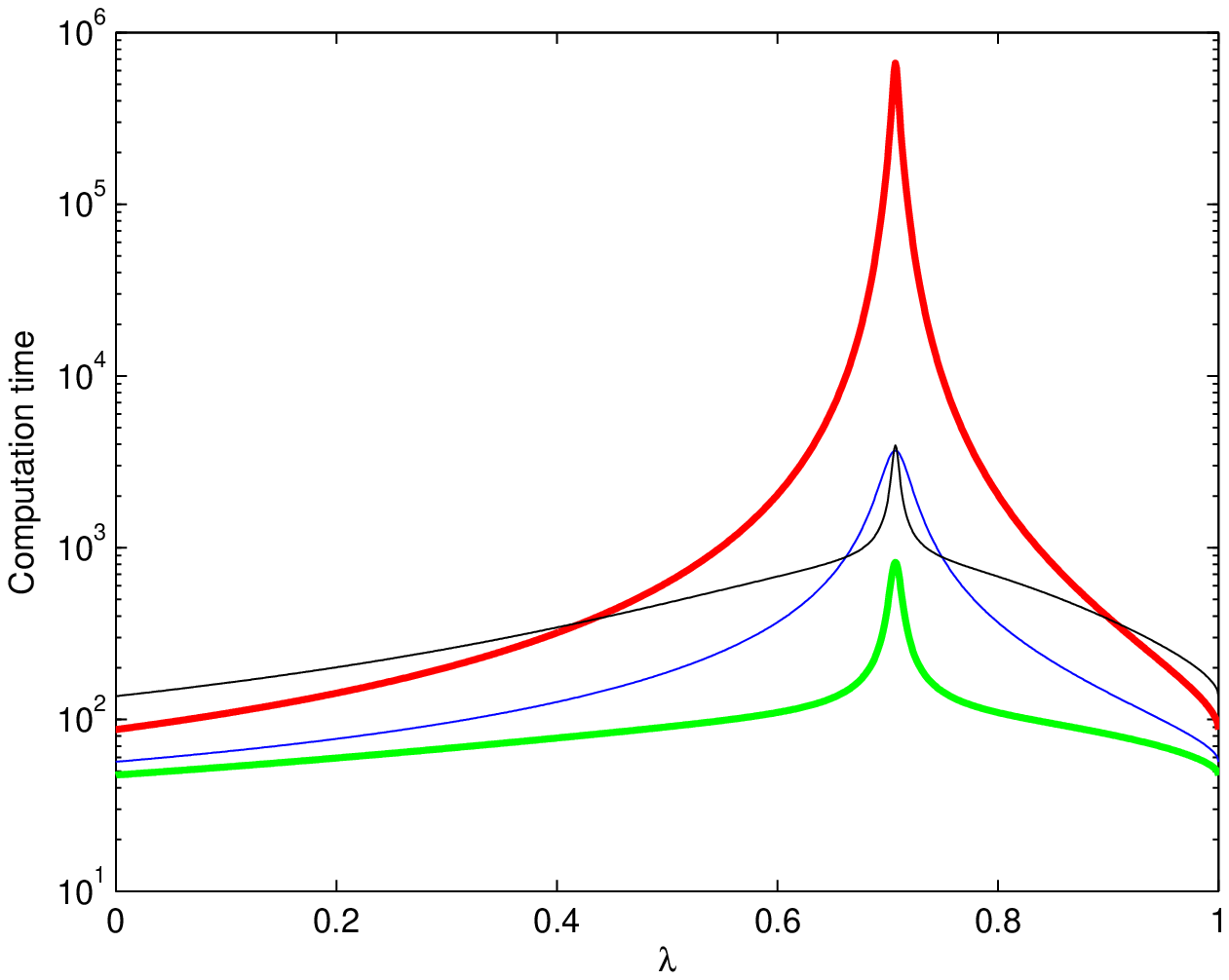}

\caption{\small Computation time against $\lambda$ for $a=0.95$. Showing time-deformed \textsc{AQC} (with $\alpha=0.5$) \; $O_1^\prime(t)$ (thick red), \; $O_1(t)$ (thick green), \; and standard \textsc{AQC} $H_0 \trr_0 H_1$ (thin blue). The black line corresponds to no time deformation, i.e. $\left[(P_x^0 \otimes 1) \trr_{\star} (1 \otimes P_x^1)\right] \trr_0 H_1$, where $\trr_{\star}$ is executed using a time variable $\min\{a, t_0\}$.}
\label{fig:time_ent1}
\end{figure}

The difference between the two time-deformed computations, $O_1(t)$ and $O_1^\prime(t)$, is evident. Figure~\ref{fig:time_ent11} illustrates this further by
showing the gap between the computation time of either schemes with respect to $\lambda$. We see that the discrepancy between the
two computations reaches a peak for a maximally entangled final ground state.

\begin{figure}[htb]
\centering
\includegraphics[width=0.65\textwidth]{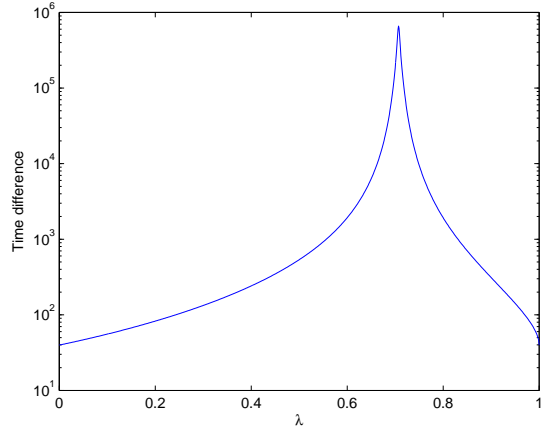}
\caption{\small Difference in computation time between the two time-deformed computations, $O_1(t)$ and $O_1^\prime(t)$. $\alpha=0.5$.}
\label{fig:time_ent11}
\end{figure}

The performance of the time-deformed algorithms for $a=0.5$ and $\alpha=0.1$ are shown in Figure~\ref{fig:time_ent12}. This time the computation $O_1^\prime(t)$ has a clear advantage over $O_1(t)$ and the standard \textsc{AQC} as the amount of entanglement reaches its peak, around $\lambda=1/\sqrt{2}$.

\begin{figure}[htb]
\centering
\includegraphics[width=0.65\textwidth]{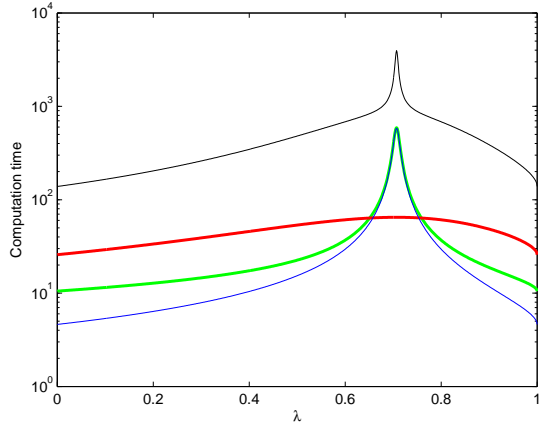}
\caption{\small Computation time against $\lambda$ for $a=0.5$. Showing time-deformed \textsc{AQC} (with $\alpha=0.1$) \; $O_1'(t)$ (thick red), \; $O_1(t)$ (thick green), \; and standard \textsc{AQC} (thin blue). The black line corresponds to no time deformation, \textit{i.e.} $\left[(P_x^0 \otimes 1) \trr_{\star} (1 \otimes P_x^1)\right] \trr_0 H_1$, where $\trr_{\star}$ is executed using a time parameter $\min\{a, t_0\}$.}
\label{fig:time_ent12}
\end{figure}

\subsection{2-\textsc{SAT} example}\label{SS:2sat}

In this example we consider a Boolean 2-\textsc{SAT} problem over four variables (qubits). We wish to find a satisfiable assignment
for the following disjunctive normal form (\textsc{DNF}) expression
\begin{equation}
\label{eq:dnf}
((x_1 \wedge x_2) \vee (\neg x_1 \wedge \neg x_2)) \wedge ((x_3 \wedge x_4) \vee (\neg x_3 \wedge \neg x_4))\enspace .
\end{equation}
There are four such assignments, namely, $(0,0,0,0)$, $(0,0,1,1)$, $(1,1,0,0)$, $(1,1,1,1)$. The final Hamiltonian corresponding
to this problem is composed of 2-local Hamiltonians. To see this we define the projector
\begin{equation}
P_z \ass \frac{1}{2}(I - Z)\enspace ,
\end{equation}
where $Z$ is the Pauli \textsc{Z} matrix. We see that:
\begin{multline}
P_z\bra x \ket = x \bra x \ket, \quad P_z \otimes P_z \bra xy \ket = (x \wedge y) \bra xy \ket,\\ \quad P_z \oplus P_z \bra xy \ket = (x \vee y) \bra xy \ket,
\quad (1 - P_z) \bra x \ket = (\neg x) \bra x \ket\enspace .
\end{multline}
An encoding of the above 2-\textsc{SAT} is then given by:
\begin{multline}
\label{eq:hprob}
H_1 = 1 - \frac{1}{2} \left[ P_z^2 + (1-P_z)^2 \right] \otimes \frac{1}{2} \left[ P_z^2 + (1-P_z)^2 \right] \\
= 1 - \frac{1}{4} \left[ P_z^4 + P_z^2 (1-P_z)^2 + (1-P_z)^2P_z^2 + (1-P_z)^4  \right]\enspace ,
\end{multline}
where we have used shorthands $P_z^2$ and $P_z(1-P_z)$ for $P_z \otimes P_z$ and for $P_z \otimes (1-P_z)$ respectively.

The problem Hamiltonian has a degenerate ground state whose basis vectors are satisfying assignments for our 2-\textsc{SAT}. Suppose there
is an oracle with a probability of $1/4$ of yielding any one of the four satisfying assignment. This oracle is represented by a black-box
Hamiltonian whose ground state is a superposition of the four assignments. Hence,
\begin{equation}
H_{\clubsuit} = 1 - \bra v \ket \kett v \brat, \quad \bra v \ket = \frac{1}{2}\left(\bra 0000 \ket + \bra 0011 \ket + \bra 1100 \ket + \bra 1111 \ket\right)\enspace .
\end{equation}
The initial Hamiltonian has a ground state $\bra 0000 \ket_x$, \textit{i.e.} a superposition of $2^4$ possible assignments. It is given by
\begin{equation}
H_0 = P_x^0 \oplus P_x^0 \oplus P_x^0 \oplus P_x^0\enspace .
\end{equation}

We describe a number of computations which solve the 2-\textsc{SAT} problem using the oracle's knowledge. These computations differ
both in the way in which the oracle's ``answers'' are processed and in their evaluation of the \textsc{SAT} clauses. The details are as follows.
In the first round of experiments we assume that $H_1$ is given to us. In this case we employ two different time-deformed tangle machines, drawn in Figure~\ref{fig:qudit}, which represent the following computations:
\begin{subequations}
\label{eq:defaqc}
\begin{equation}
O_1(t) = (H_0 \trr_1 H_{\clubsuit}) \trr_0 H_1\enspace \;
\end{equation}
\begin{equation}
O_1^\prime(t) = (H_0 \trr_1 H_1) \trr_0 (H_{\clubsuit} \trr_1 H_1)\enspace .
\end{equation}
\end{subequations}

\begin{figure}[htb]
\centering
\psfrag{0}[c]{\small $H_0$}
\psfrag{1}[c]{\small $H_{\clubsuit}$}
\psfrag{2}[c]{\small $H_1$}
\psfrag{a}[c]{\small $O_1$}
\psfrag{b}[c]{\small $O_2$}
\psfrag{c}[c]{\small $O_3$}
\psfrag{t}[c]{\small $t_1$}
\psfrag{s}[c]{\small $t_0$}
\psfrag{e}[c]{\small $O_1^\prime$}
\psfrag{f}[c]{\small $O_2^\prime$}
\psfrag{g}[c]{\small $O_3^\prime$}
\includegraphics[width=0.4\textwidth]{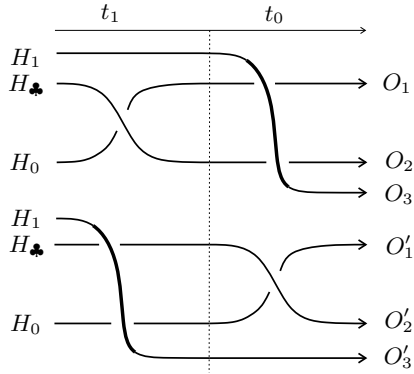}
\caption{\small Equivalent computations.}
\label{fig:qudit}
\end{figure}

In another round of experiments we embed $H_1$ into the tangle machine. Hence, instead of using the expression for $H_1$ in~\eqref{eq:hprob}, we consider
$H_1$ to be computed during the evolution of the system. In other words, we let:
\begin{equation}
G(t) = \left[(1-P_z^4) \trr_{\star} (1-P_z^2(1-P_z)^2)\right] \trr_{\star} \left[(1-(1-P_z)^2P_z^2) \trr_{\star} (1-(1-P_z)^4)\right]\enspace ,
\end{equation}
where $a \trr_{\star} b \ass (1-\frac{1}{2} t) a + \frac{1}{2}t b$. Note that indeed
\[
\lim_{t \to 1} G(t) = H_1\enspace .
\]

Substituting $H_1$ with its evolving version $G(t)$ in \eqref{eq:defaqc} yields
\begin{subequations}
\label{eq:defaqc1}
\begin{equation}
O_1(t) = (H_0 \trr_1 H_{\clubsuit}) \trr_0 G(t)\enspace \;
\end{equation}
\begin{equation}
O_1'(t) = (H_0 \trr_1 G(t)) \trr_0 (H_{\clubsuit} \trr_1 G(t))\enspace .
\end{equation}
\end{subequations}
The diagrams describing the two computations are shown in Figure~\ref{fig:qudit1}. The dashed region represents the logic underlying $G(t)$.
Note that the number of inputs/outputs of this \emph{submachine} corresponds to the number of conjunctive clauses in the original \textsc{DNF} \eqref{eq:dnf}.

\begin{figure}[htb]
\centering
\psfrag{0}[c]{\small $H_0$}
\psfrag{1}[c]{\small $H_{\clubsuit}$}
\psfrag{2}[r]{\small $_{1-P_z^4}$}
\psfrag{3}[r]{\small $_{1-P_z^2(1-P_z)^2}$}
\psfrag{4}[r]{\small $_{1-(1-P_z)^2P_z^2}$}
\psfrag{5}[r]{\small $_{1-(1-P_z)^4}$}
\psfrag{a}[c]{\small $O_1$}
\psfrag{b}[c]{\small $O_2$}
\psfrag{c}[c]{\small $O_3$}
\psfrag{t}[c]{\small $t_1$}
\psfrag{s}[c]{\small $t_0$}
\psfrag{e}[c]{\small $O_1^\prime$}
\psfrag{f}[c]{\small $O_2^\prime$}
\psfrag{h}[c]{\small $O_3^\prime$}
\psfrag{g}[c]{\small $G(t)$}
\includegraphics[width=0.95\textwidth]{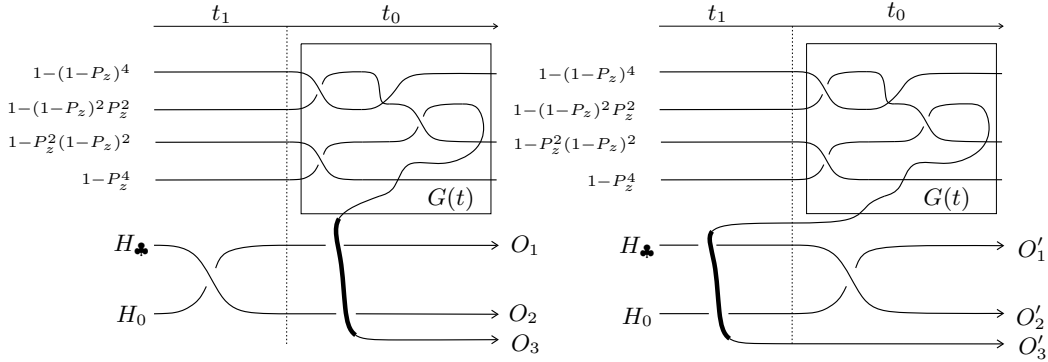}
\caption{\small Equivalent computations in which the problem Hamiltonian logic is encoded within the tangle (boxed region).}
\label{fig:qudit1}
\end{figure}

Figure~\ref{fig:qudit2} shows the energy gap over time for the various modes of computation in this example. Here we use $\alpha=0.5$.

\begin{figure}[htb]
\centering
\includegraphics[width=0.65\textwidth]{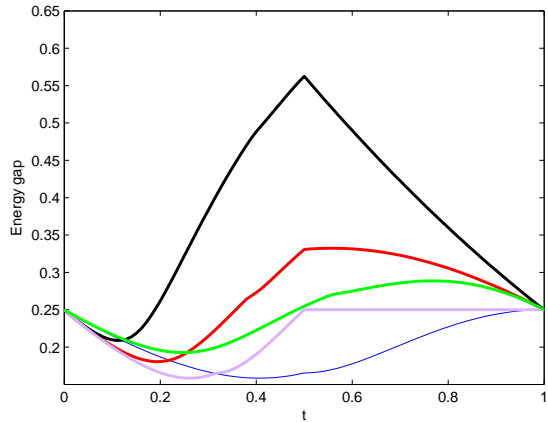}
\caption{\small Energy gaps in the qudit example. Standard \textsc{AQC} using $H_1$ (thin blue), \; standard \textsc{AQC} using $G(t)$ (green), \; time-deformed \textsc{AQC} using $H_1$ (purple), \; time-deformed \textsc{AQC} $O_1(t)$ using $G(t)$ (red), \; time-deformed \textsc{AQC} $O_1^\prime(t)$ using $G(t)$ (black).}
\label{fig:qudit2}
\end{figure}





\section{Utility of tangle machines}\label{S:Utility}

In this paper we have introduced tangle machines and they have told us tales of information flow in networks in both classical and quantum settings. In what way can this distributive structure be of use to information theorists and to computer scientists? We present two advantages that such a topological description provides following \cite{CarmiMoskovich:15c} and \cite{CarmiMoskovich:14}.

\subsection{Topological invariants}\label{SS:Invariants}

Knot diagrams are planar projections of knots in $\mathds{R}^3$ equipped with over/under information. Knots in $3$--space are ambient isotopic if and only if any two diagrams that represent them are related by a finite sequence of Reidemeister moves. Is there an analogous topological interpretation for tangle machines?

We have shown that tangle machines arise as planar projections of objects we call \emph{inca foams}~\cite{CarmiMoskovich:15c}, that are necklaces of spheres glued together along disjoint discs, all jointly embedded in $\mathds{R}^4$. Two inca foams are ambient isotopic if and only if any two tangle machines which represent them are related by a finite sequence of Reidemeister moves. This topological interpretation provides us tools to construct and to study \emph{topological invariants} of tangle machines, that are characteristic quantities which coincide when tangle machines are equivalent. These will ideally have clear and compelling interpretations in terms of information. We mention two invariants.

\begin{figure}
\centering
\includegraphics[width=0.8\textwidth]{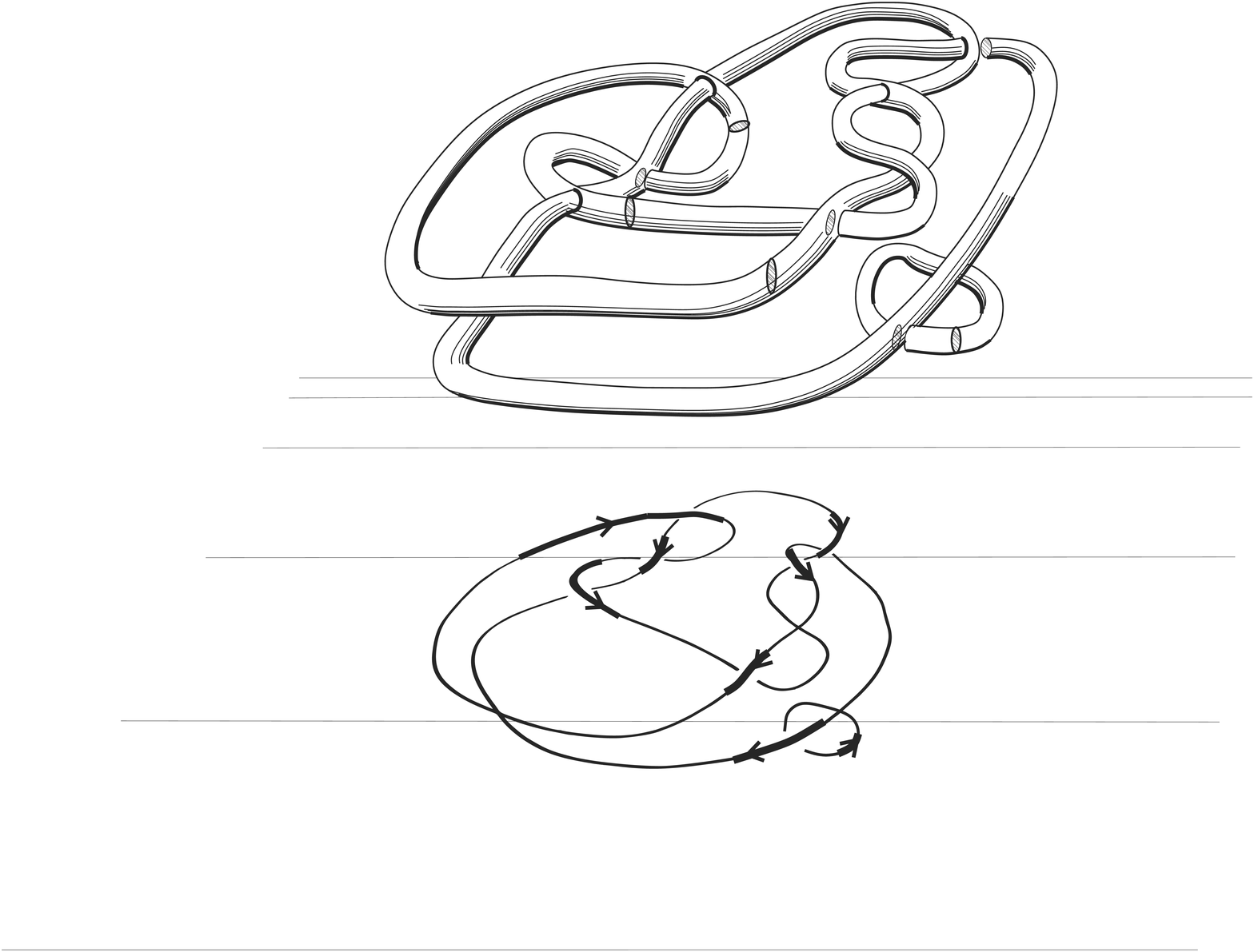}
\caption{\label{F:NecklaceExample} An Inca foam and a tangle machine representing it.}
\end{figure}

\begin{description}
\item[Shannon capacity] There is a well-defined notion of \emph{Shannon capacity} for a graph which measures the Shannon capacity of an communications channel defined by the graph (\textit{e.g.}~\cite{Erickson:14}). This notion extends to tangle machines~\cite{CarmiMoskovich:15c}.

Define an assignment of colours to arcs of a tangle machine $M$ to be \emph{confusable} is there is one of these colours that is uniquely determined by the others. Let $\textrm{Cap}_k(M)$ denote the number of distinct assigments of $k$ colours to arcs of $M$ which $M$ admits. Now define:

\begin{equation}
\mathrm{Cap}(M)\ass \sup_{k\in\mathds{N}}\sqrt[k]{\mathrm{Cap}_k(M)}
\end{equation}

\item[Complexity] Knots admit a binary operation called \emph{connect~sum} under which they form a commutative monoid:

\[
\includegraphics[height=45pt]{trefoilcol}\quad \raisebox{20pt}{\Hash} \quad \includegraphics[height=45pt]{trefoilcol}\quad \raisebox{20pt}{$=$} \quad \includegraphics[height=45pt]{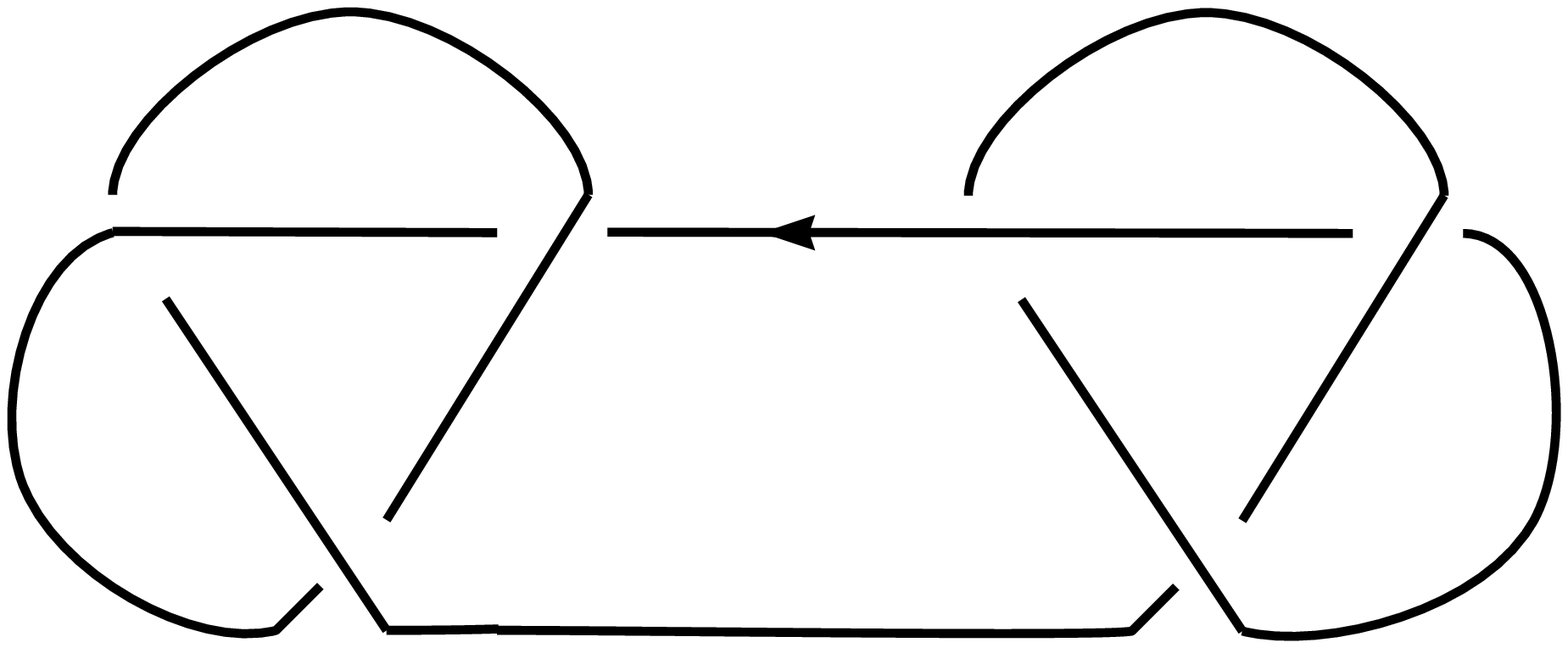}
\]

\begin{figure*}
\renewcommand{\thesubfigure}{\Alph{subfigure}}
\begin{subfigure}[t]{0.42\linewidth}
\psfrag{x}[c]{}
\psfrag{y}[c]{}
\psfrag{a}[c]{}
\psfrag{b}[c]{}
    \centering
    \vspace{10pt}
    \includegraphics[width=0.9\textwidth]{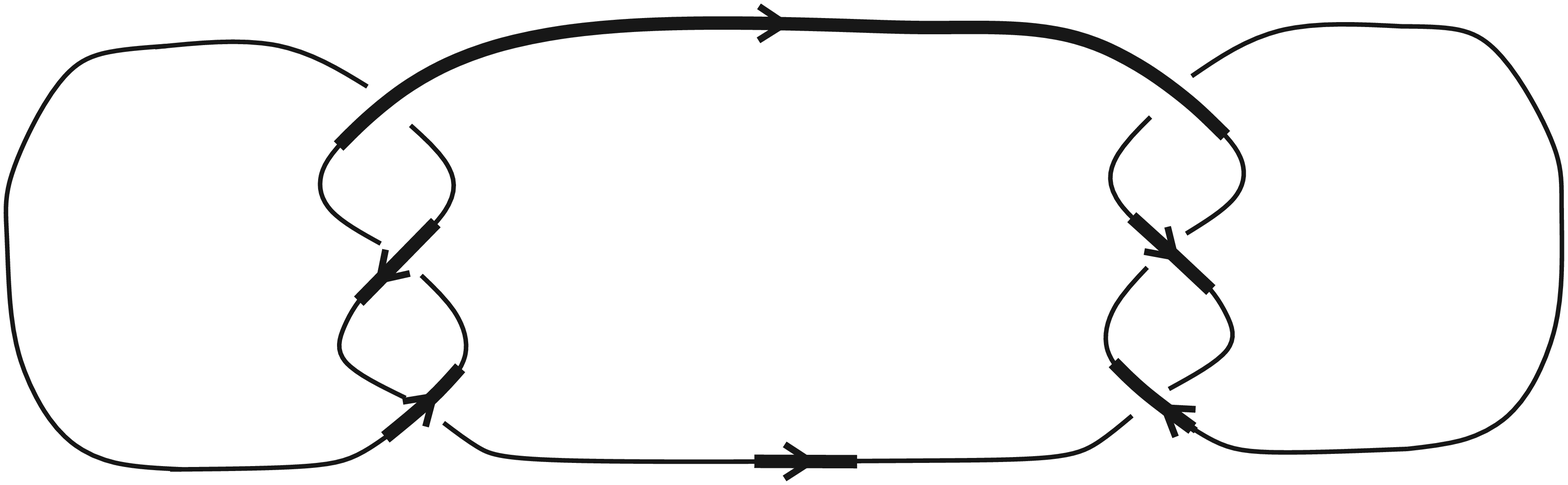}
\caption{\label{F:square1}}
\end{subfigure}
\qquad
\begin{subfigure}[t]{0.42\linewidth}
\psfrag{x}[c]{}
\psfrag{y}[c]{}
\psfrag{a}[c]{}
\psfrag{b}[c]{}
    \centering
    \vspace{10pt}
    \includegraphics[width=0.9\textwidth]{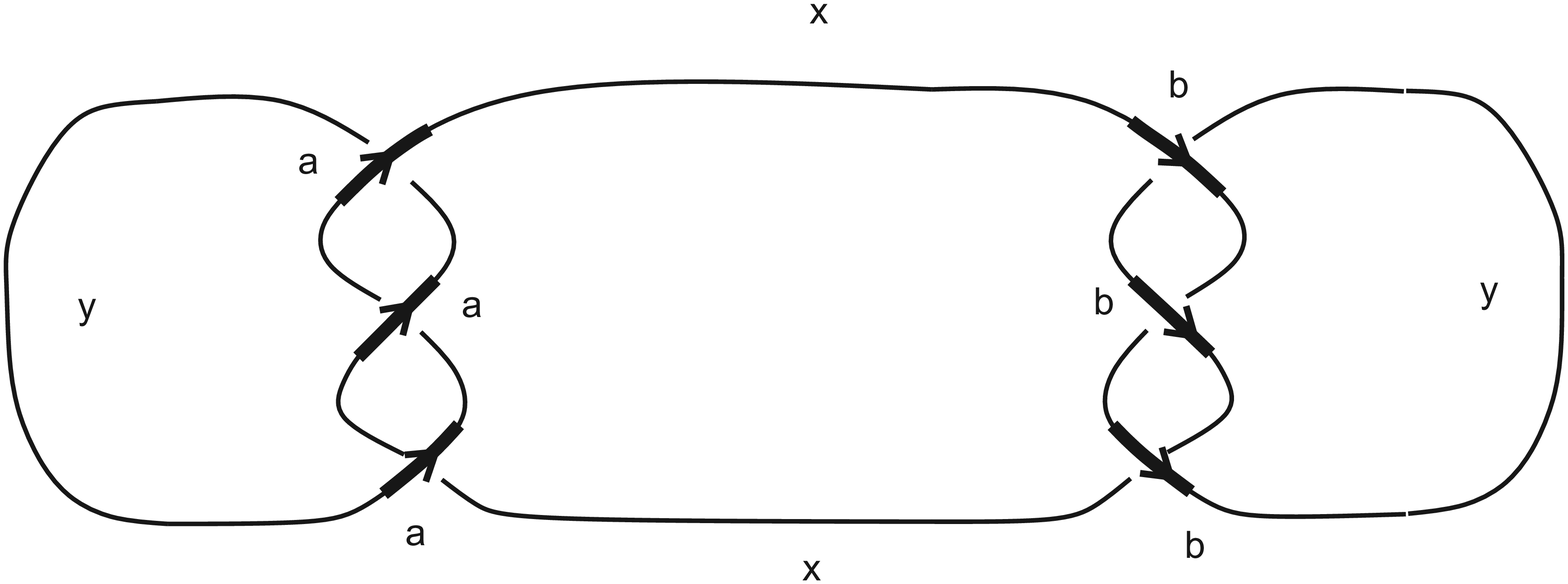}
\caption{\label{F:square2}}
\end{subfigure}%
\caption{\label{F:Complexity} \small Machine \ref{F:square1} has complexity $1$. Machine~\ref{F:square2} has complexity $2$ if the horizontal strands are the same colour, otherwise it too has complexity $1$.}
\end{figure*}

Coloured knots also form a commutative monoid under the connect-sum operation~\cite{Moskovich:06}. Although many related diagrammatic algebraic structures do not~\cite{Kim:00,KishinoSatoh:04}, tangle machines also form a commutative monoid under the connect-sum operation~\cite{CarmiMoskovich:15c}. More generally, we may define two submachines as being \emph{independent} if they are connected along a set of arcs all of which share the same colour~\cite{CarmiMoskovich:14b}. The maximal number of nontrivial connect-summands is a topological invariant of tangle machines, and is a (rough) measure of tangle machine complexity. See Figure~\ref{F:Complexity}.
\end{description}

\subsection{Flexibility of description}\label{SS:FaultTolerance}

A second advantage of the tangle machine formalism is that, unlike directed graphs, it is flexible. Reidemeister moves give different machines representing different information flows which globally achieve the same task.

In Section~\ref{S:Quantum} we will give examples in which the performance of one machine is significantly better than the performance of an equivalent machine  although both achieve the same task. Further examples are in~\cite{CarmiMoskovich:15b}. Here, following~\cite{CarmiMoskovich:14}, we shall present an idea for an application of the tangle machine description to fault tolerance of sensor networks which perform fusion using covariance intersection.

We may imagine a realization of a tangle machine using with modular robots called `cubelets' as shown in Figure~\ref{F:cubelets}.

\begin{figure}
\centering
\psfrag{a}[c]{$y$}
\psfrag{b}[c]{$x$}
\psfrag{c}[l]{$z$}
\includegraphics[width=0.9\textwidth]{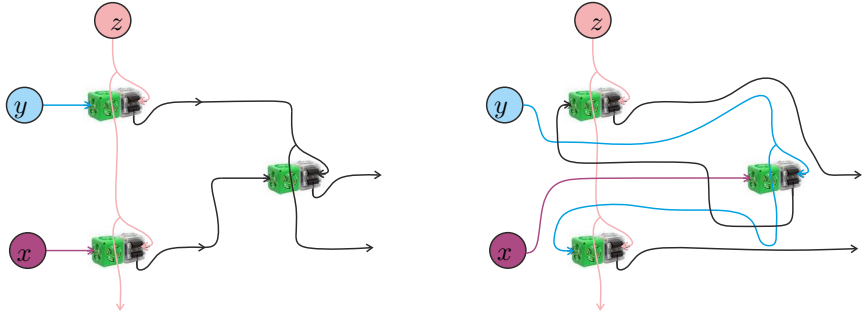}
\caption{\label{F:cubelets} \small The R3 move realized using cubelets~{\cite{Cubelets}}. Both these networks have the same outputs and thus represent different realizations of essentially the same modular robot. The cubelets themselves do not move, and what changes is only the configuration of which cubelet sends its output to which other cubelet.}
\end{figure}

 Three transmitting devices continuously stream data $x(T)$, $y(T)$, and $z(T)$. Our task is to combine these data streams, eliminating redundancy (\textit{e.g.} because of the Problem of Double~Counting~\cite{Jazwinski:70}). Two schemes to combine the data streams at a fixed time $T\geq 0$ are represented by the left and the right hand side of the R3 move in Figure~\ref{F:ReidemeisterMoves1}.

In the left machine, combine $x(T)$ and $y(T)$ with $z(T)$ with parameter $t\in(0,1)$ to obtain fused data streams $x(T)\trr_t z(T)$ and $y(T)\trr_t z(T)$. Assume that all estimators are Gaussian and that the fusion is performed using covariance intersection. We read the fusion from bottom to top because the overcrossing arc is directed from left to right. If it were oriented from right to left, we would read from top to bottom and we would be filtering out the data stream $z(T)$ from $x(T)\trr_t z(T)$ and from $y(T)\trr_t z(T)$. On the right, $x(T)\trr_t z(T)$ is combined with $y(T)\trr_t z(T)$ using a weight $s\in(0,1)$.

 The left and the right machine describe equivalent data stream fusion schemes, which is visually indicated by the fact that the diagrams are related by sliding one overcrossing arc over another. Indeed, the result of the data stream fusion in the right machine is $\left(x(T)\trr_s y(T)\right)\trr_t z(T)$, which is the same combined data stream as in the left machine because redundant double~appearance of $z(T)$ in $x(T)\trr_t z(T)$ and $y(T)\trr_t z(T)$ is eliminated by $\trr_s$ by assumption.

 However, there is an important difference between these two schemes on the left and right of the R3 move. Imagine at some time $T_1>0$ that stream $y$ becomes faulty. In this case the left machine is superior because it contains the intermediate data stream $x(T)\trr_t z(T)$ which might be useful even when $\left(x(T)\trr_t z(T)\right)\trr_s\left(y(T)\trr_t z(T)\right)$ is junk. Conversely, if $z$ becomes faulty at some time $T_2>0$ then the right machine would be preferred. The top overcrossing arc might slide back and forth at different times. Thus the machines involved in the R3 move might be describing the underlying logic of a simple fault-tolerant data stream fusion network.

A more general version of the same scheme is described in Figure~\ref{fig:topo}. Consider the network illustrated in the upper left corner of Figure~\ref{fig:topo}. This network has a set of outputs outside the bounding disk (not shown). Some set of intermediate edges lies inside the subnetwork designated by $N_L$ and represented as an empty circle. In the course of network operation, erroneous streams of data cause one or more of the edges $0$,$1$, and $2$, to carry faulty pieces of information, \textit{e.g.} biased, inconsistent and otherwise unreliable estimates. This is detected inside the network. To inhibit the influence of this contamination on edges within $N_L$, the network performs Reidemeister moves on itself, transforming itself into an equivalent network. This is achieved by `sliding' the faulty edges all the way over $N_L$, by repeated application of the second and third Reidemeister moves. In the resulting topology, the faulty edges have no effect on $N_L$, and so local costs are improved.

\begin{figure*}[t]
\centering
\psfrag{s}[c]{\small $N_L$}
\psfrag{a}[c]{\small \emph{fault} $000$}
\psfrag{b}[c]{\small \emph{fault} $00X$}
\psfrag{c}[c]{\small \emph{fault} $0X0$}
\psfrag{d}[c]{\small \emph{fault} $X00$}
\psfrag{e}[c]{\small \emph{fault} $X0X$}
\psfrag{f}[c]{\small \emph{fault} $0XX$}
\psfrag{0}[c]{\small $0$}
\psfrag{1}[c]{\small $1$}
\psfrag{2}[c]{\small $2$}
\includegraphics[width=0.65\textwidth]{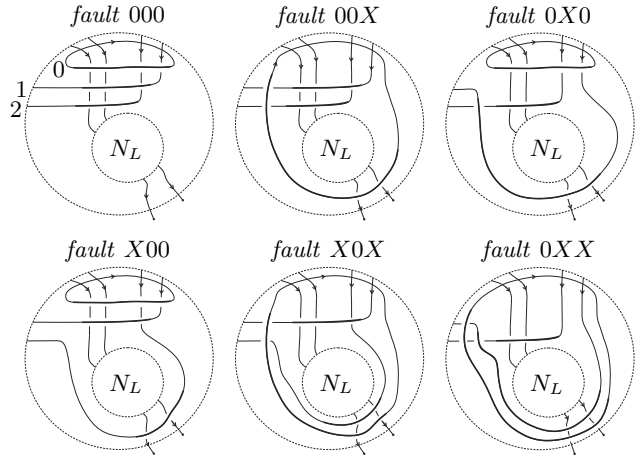}
\caption{\small Topological fault-tolerance. Switching between equivalent networks ameliorate the effects of faults. The fault code appearing above each network should be read as follows. The digits from left to right correspond to edges $2$,$1$ and $0$. An ``X'' in the place of the $i$th
digit signifies faulty content in its respective edge.}
\label{fig:topo}
\end{figure*}


In the above example, things were as easy as sliding faulty edges forward, but this is only because the networks we considered are topologically uninteresting. In the general case an optimal configuration may be difficult to identify.

\section{To learn a tangle machine}\label{S:Learn} 

Let's now try a different way to connect tangle machines to the real world. Given a data set we may ask whether there exists a tangle machine which adequately the flow of information such a set. Can a tangle machine represent a data-generating process for real world-data?

In this section we identify an interaction within real world data. It may be possible to expand our method and to adapt existing causality-detection algorithms to construct tangle machines from data.

\subsection{Bayesian networks and Granger causality}

One autoregressive system $X=\{X(t)\, \mid t\in \mathds{N}\}$ \emph{Granger causes} another autoregressive system $Y=\{Y(t)\, \mid t\in \mathds{N}\}$ if lagged values $\mathcal{L}^{T-p: T-1}(X)=(X(T-1),X(T-2),\ldots, X(T-p))$ contain information about $Y_T$ beyond the information contained in past values of $Y$ itself for some fixed $p\in \mathds{N}$ and for all $T>p$. Magnitude of Granger causality is measured a quantity called the \emph{transfer entropy} from $X$ to $Y$.

A Bayesian network is a directed acyclic graph whose nodes represent random variables and whose edges represent conditional dependencies. Generalizing, nodes of Bayesian networks may represent autoregressive systems and edges may represent transfer entropies from the tail to the head of each edge~\cite{Pearl:09}. Thus the tail of an edge Granger causes its head. There is a well-developed family of algorithms to learn transfer entropies for a given Bayesian network from time-series data.

We consider an interaction with a single input patient as representing a triangular Bayesian network. The agent corresponds to the source in the network, and the two patients correspond to the remaining two nodes. Thus, our idea is that the agent Granger causes both patients.  See Figure~\ref{fig:Bayesian}.

Our model is:
\begin{equation}
\begin{array}{l}
Z_{k+1} = s_{X \to Z} \mathcal{L}^{k-p:k-1}(X) + s_{Y \to Z} \mathcal{L}^{k-p:k-1}(Y)+\nu_Z \\
X_{k+1} = s_{Z \to X} \mathcal{L}^{k-p:k-1}(Z) + s_{Y \to X} \mathcal{L}^{k-p:k-1}(Y)+\nu_X \\
Y_{k+1} = s_{X \to Y} \mathcal{L}^{k-p:k-1}(X) + s_{Z \to Y} \mathcal{L}^{k-p:k-1}(Z)+\nu_Y\enspace ,
\end{array}
\end{equation}
\noindent where all variables are assumed to be normally distributed, $s$ terms are constant, and $\nu$ terms are noise.

Under the above assumptions, a short computation shows that the expression for the transfer entropies is maximized when the $s$ parameters are maximal \cite{Carmi:13}.

Thus to maximize the transfer entropy for a given data set, our goal becomes to solve the following least-square problem to estimate the $s$ parameters:

\begin{equation}
\begin{array}{l}
Z_{k+1} = s_{X \to Z} \mathcal{L}^{k-p:k-1}(X) + s_{Y \to Z} \mathcal{L}^{k-p:k-1}(Y) \\
X_{k+1} = s_{Z \to X} \mathcal{L}^{k-p:k-1}(Z) + s_{Y \to X} \mathcal{L}^{k-p:k-1}(Y) \\
Y_{k+1} = s_{X \to Y} \mathcal{L}^{k-p:k-1}(X) + s_{Z \to Y} \mathcal{L}^{k-p:k-1}(Z)\enspace .
\end{array}
\end{equation}
We then choose the configuration that corresponds to the greater value among each pair of parameters $\{s_{X\to Y}, s_{Y\to X}\}$, $\{s_{X\to Z}, s_{Z\to X}\}$, and $\{s_{Z\to Y}, s_{Y\to Z}\}$. This process fits a single input-patient interaction to a triangular Bayesian network.

By extension, an interaction with agent $Y$, input patients $X_1,X_2,\ldots, X_k$, and output patients $Z_1,Z_2,\ldots,Z_k$ with $Z_i= X_i\trr_\omega Y$  corresponds to a Bayesian network with nodes $Y,X_1,X_2,\ldots,X_k,Z_1,Z_2,\ldots Z_k$ in which the $X_i's$ arise as mixtures of autoregressive processes such that $Y$ only acts on one of these and is independent of the others. The extension of our method to these and to more complicated Bayesian networks will be considered in future work.

\begin{figure}
\begin{minipage}{1.1in}
\psfrag{a}[c]{$Y$}\psfrag{c}[c]{$X$}\psfrag{b}[c]{$Z$}
\scalebox{0.5}{\includegraphics[width=2.2in]{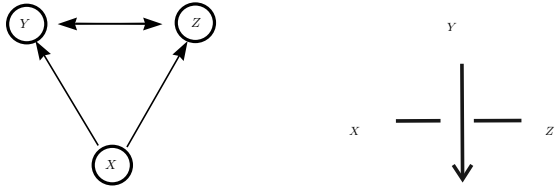}}
\end{minipage}
\qquad\qquad\quad
\begin{minipage}{1.1in}
\psfrag{u}[r]{\small $X$}
\psfrag{y}[c]{\small $Y$}
\psfrag{x}[l]{\small $Z$}
\scalebox{0.5}{\includegraphics[width=2.2in]{cross.eps}}
\end{minipage}
\caption{\label{fig:Bayesian} An interactions and its Bayesian network analogue. The arrow between $Y$ and $Z$ may point either way. An interaction-learning algorithm should learn the correct network configuration. Thus it should learn which the agent is, and the quandle operation.}
\end{figure}


We present an example in which we have learnt an interaction from real-life search data. Around 2006, the field of signal processing was revolutionized by the appearance of the concept of \emph{compressed sensing}. Compressed sensing recovers signals using far fewer observations than required by the Shannon--Nyquist sampling theorem. Several people were involved, one of whom was Emmanuel Candes. This story may be described by an interaction of the form: E.~Candes \emph{causes} a shift of state from \emph{Shannon--Nyquist sampling} to \emph{compressed sensing}.

We used \emph{Google Trends} to generate time-series data of hits per year from 2000 to 2010 for the three search terms: {\tt E.~Candes}, {\tt Shannon--Nyquist
sampling}, and {\tt compressed sensing}. We then used the filtering-based detection algorithm described in~\cite{Carmi:13} to recover the underlying Bayesian
network. As shown in Figure~\ref{fig:detection}, the actual relationship between the time-series could be recovered up to a time order of {\tt Shannon--Nyquist sampling} and {\tt compressed sensing}. There was an ambiguity between two possible interactions both of which identify {\tt E.~Candes} as a cause of the paradigm shift.

A refined and expanded version of this method will be needed to detect machines with more than one interaction.

\begin{figure}[htb]
\centering
\psfrag{t}[c]{\small \emph{time series data}}
\includegraphics[width=0.95\textwidth]{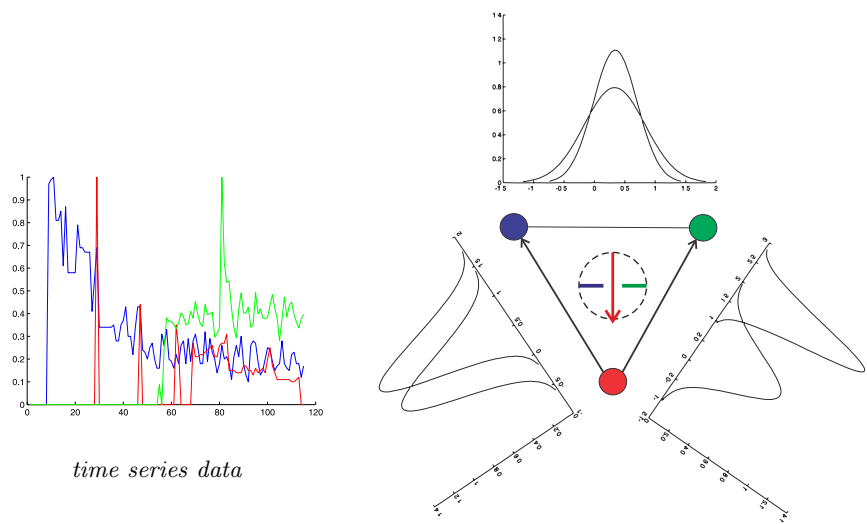}
\caption{\small Detecting an interaction from time-series data. Colour code: Red = {\tt E.~Candes}, \; Blue = {\tt Shannon--Nyquist sampling}, \; Green
  = {\tt compressed sensing}. The Bell-curves show the statistical significance of determining the corresponding arrow direction. Roughly, the further apart the curves are for each arrow the more significant is its (causal) direction.}
\label{fig:detection}
\end{figure}

\bibliographystyle{amsplain}

\newpage
\appendix

\section{Appendix}

The proof of Proposition~\ref{P:Geodesic} is the combination of the following two lemmas which are surely well-known.

\begin{lem}
The covariance intersection of weight $\omega\in [0,1]$ of two unnormalized Gaussian probability density functions $p(X)$ and $p(Y)$ is $p(X)^{1-\omega}p(Y)^\omega$.
\end{lem}

\begin{proof}

Let $p(X), p(Y)$ be unnormalized Gaussian densities parameterized by their mean and covariance, $(\hat{X}, C_X)$ and $(\hat{Y}, C_Y)$, respectively.
The pdf $p(X)^{1-\omega}p(Y)^\omega$ is given by the expression:
\begin{equation}
\label{eq:zz}
\exp\left(-\frac{1}{2} \left[ (1-\omega) (Z - \hat{X})^T C_X^{-1} (Z - \hat{X}) + \omega (Z - \hat{Y})^T C_Y^{-1} (Z - \hat{Y})\right]\right)\enspace .
\end{equation}

We verify that this expression equals:
\begin{equation}
\exp\left(-\frac{1}{2} (Z - \hat{Z})^T C_Z^{-1} (Z - \hat{Z}) \right) \enspace ,
\end{equation}
for:
\begin{subequations}
\begin{equation}\label{E:21}
\hat{Z} = (1-\omega) C_Z C_X^{-1} \hat{X} + \omega C_Z C_Y^{-1} \hat{Y}\enspace \;
\end{equation}
\begin{equation}\label{E:22}
C_Z^{-1} = (1-\omega) C_X^{-1} + \omega C_Y^{-1}\enspace .
\end{equation}
\end{subequations}

We compute that $C_X^{-1} = -\frac{\partial^2 \log p(x)}{\mathop{\partial x} \mathop{\partial x}^T}$ which is called the \emph{precision matrix} of $p(X)$. In particular,
\begin{equation}
C_Z^{-1} = - \frac{\partial^2 \log p(z)}{\mathop{\partial z} \mathop{\partial z}^T} = (1-\omega)\, C_x^{-1} + \omega\, C_Y^{-1} \enspace ,
\end{equation}
which follows from~\eqref{eq:zz}.

The mode $\hat{X}$ satisfies:
\begin{equation}
 \left.\frac{\partial \log p(x)}{\partial x} \right|_{x=\hat{X}} = 0\enspace .
\end{equation}
The mode $\hat{Z}$ of $p(Z)$ is thus obtained from~\eqref{eq:zz} by
\begin{equation}
2 \left.\frac{\partial \log p(z)}{\partial z} \right|_{z=\hat{Z}} = (1-\omega)\, C_x^{-1} (\hat{Z} - \hat{X}) + \omega\, C_Y^{-1} (\hat{Z} - \hat{Y}) = 0\enspace ,
\end{equation}
finally giving
\begin{equation}
C_z^{-1} \hat{Z} = (1-\omega) C_x^{-1} \hat{X} + \omega C_Y^{-1} \hat{Y} \enspace .
\end{equation}
\end{proof}

\begin{lem}
Let $P$ be a space of probability density functions (pdfs) over some domain. Consider the functional $J\colon\, P \to \mathds{R}$ defined by
\begin{equation}
\label{eq:slj}
J(g(X)) \ass (1-\omega) \mathrm{KL}\left( g(X) \| p(X) \right) + \omega \mathrm{KL}\left( g(X) \| q(X) \right)
\end{equation}
where $g(X)$, $p(X)$, and $q(X)$ are pdfs in $P$ and $\mathrm{KL}(\cdot \| \cdot)$ is the Kullback--Leibler divergence. Then the following holds for any $\omega \in [0,1]$ and $g(X) \in P$:
\[
J\left(\frac{p(x)^{1-\omega} q(x)^\omega}{\int_x p(x)^{1-\omega} q(x)^\omega \mathop{dx}}\right) \leq J(g(X))\enspace .
\]
\end{lem}

\begin{proof}
Writing down the Kullback--Leibler divergence explicitly, we have:
\begin{multline}
\label{eq:gg}
\min_{g(X)} J(g(X)) =
\min_{g(X)} \left[(1-\omega) \int_x g(x) \log \frac{g(x)}{p(x)} \mathop{dx} + \omega \int_x g(x) \log \frac{g(x)}{q(x)} \mathop{dx}\right] =\\
= \min_{g(X)} \int_x g(x) \log \frac{g(x)}{p(x)^{1-\omega} q(x)^\omega} \mathop{dx}= \\ =
\min_{g(X)} \left[\int_x g(x) \log \frac{g(x)}{p(x)^{1-\omega} q(x)^\omega} \mathop{dx} + \int_x p(x)^{1-\omega} q(x)^\omega \mathop{dx} - \int_x p(x)^{1-\omega} q(x)^\omega \mathop{dx} \right] = \\=
\min_{g(X)} \left[\mathrm{KL}\left( g(X) \bigg \| \frac{p(X)^{1-\omega} q(X)^s}{\int_x p(x)^{1-\omega} q(x)^s \mathop{dx}} \right) \right] - \int_x p(x)^{1-\omega} q(x)^\omega \mathop{dx} \enspace .
\end{multline}
But $\mathrm{KL}(\cdot \| \cdot) \geq 0$ vanishes if and only if its arguments coincide, and the result follows.
\end{proof}

\end{document}